\newif\ifusesiam
\IfFileExists{siamart220329.cls}{%
  \usesiamtrue
  \documentclass[final]{siamart220329}%
}{%
  \usesiamfalse
  \documentclass[11pt]{article}
  \usepackage[margin=1in]{geometry}
}

\usepackage{amsmath,amssymb,amsfonts,bm,mathtools}
\ifusesiam\else
\usepackage{amsthm}
\fi
\usepackage{graphicx}
\usepackage{booktabs}
\usepackage{array}
\usepackage{placeins}
\usepackage{hyperref}
\usepackage[numbers,sort&compress]{natbib}

\makeatletter
\@ifundefined{cref@old@refstepcounter}{}{%
  \def\refstepcounter@optarg[#1]#2{%
    \cref@old@refstepcounter{#2}%
    \cref@constructprefix{#2}{\cref@result}%
    \@ifundefined{cref@#1@alias}%
      {\def\@tempa{#1}}%
      {\def\@tempa{\csname cref@#1@alias\endcsname}}%
    \protected@edef\cref@currentlabel{%
      [\@tempa][\arabic{#2}][\cref@result]%
      \csname p@#2\endcsname\csname the#2\endcsname}}%
}%
\makeatother

\graphicspath{{figures/}}
\newcommand{\dd}{\,\mathrm{d}}
\newcommand{\ii}{\mathrm{i}}

\newcommand{\calH}{\mathcal{H}}
\newcommand{\calR}{\mathcal{R}}

\newcommand{\D}{\mathrm{D}}
\newcolumntype{L}[1]{>{\raggedright\arraybackslash}p{#1}}

\ifusesiam
\else

\newtheorem{proposition}{Proposition}
\newtheorem{corollary}{Corollary}
\newenvironment{MSCcodes}{\par\noindent\textbf{MSC codes.}\ }{\par}
\fi

\providecommand{\orcid}[1]{\ignorespaces}

\begin{document}

\title{Calibrated Pressure-Observable Born and Hessian Actions for Quantum-Assisted Waveform Inversion}

\ifusesiam
\author{Guanyu Li\thanks{School of Mathematics, Jilin University,
Changchun 130012, Jilin, P.R. China; Shenzhen Loop Area Institute (SLAI),
Shenzhen, P.R. China (\email{gyli24@mails.jlu.edu.cn}).}
\and
Jiwei Jia\thanks{School of Mathematics, Jilin University,
Changchun 130012, Jilin, P.R. China; Shenzhen Loop Area Institute (SLAI),
Shenzhen, P.R. China (\email{jiajiwei@jlu.edu.cn}, corresponding author).}
\and
Yu Wang\thanks{Key Laboratory of Digital Technology in Medical Diagnostics
of Zhejiang Province, Hangzhou, Zhejiang, P.R. China
(\email{flimanadam@gmail.com}).}
\and
Yuping Duan\thanks{School of Mathematical Sciences, Beijing Normal University,
Beijing 100875, P.R. China (\email{doveduan@gmail.com}).}}
\else
\author{Guanyu Li\textsuperscript{a,b}
\and
Jiwei Jia\textsuperscript{a,b}
\and
Yu Wang\textsuperscript{c}
\and
Yuping Duan\textsuperscript{d}\\[0.5em]
\textsuperscript{a}School of Mathematics, Jilin University, Changchun 130012, Jilin, P.R. China\\
\textsuperscript{b}Shenzhen Loop Area Institute (SLAI), Shenzhen, P.R. China\\
\textsuperscript{c}Key Laboratory of Digital Technology in Medical Diagnostics of Zhejiang Province, Hangzhou, Zhejiang, P.R. China\\
\textsuperscript{d}School of Mathematical Sciences, Beijing Normal University, Beijing 100875, P.R. China\\
Corresponding author: \texttt{jiajiwei@jlu.edu.cn}}
\fi

\maketitle

\begin{abstract}
We construct a pressure-consistent operator-and-readout interface for Born,
adjoint, and Gauss--Newton actions in constant-density acoustic full-waveform
inversion (FWI) using Schr\"odingerised propagation. The energy variables
\(\pi=c^{-1}\partial_tu\) and \(\bm q=\nabla u\) give an auxiliary-space
Hamiltonian, while physical pressure \(p=c\pi\) depends explicitly on
wavespeed. Its directional derivative
\(\D(c\pi)[c_0](\delta c)=c_0\delta\pi+\delta c\,\pi_0\)
contains both a propagated wavefield sensitivity and a direct
receiver-calibration term. We combine Duhamel differentiation with the
receiver-row derivative and retain both contributions in the Born map, its
adjoint, and the Gauss--Newton normal action. We prove a conditional
consistency estimate with a periodic second-order finite-difference
specialization. A resource model accounts for coefficient access,
state preparation, input loading, LCU normalization, quadrature, and
selected-output measurement. A compiled
nine-qubit instance realizes structured preparation, product-formula
propagation, a derivative-LCU block, and calibrated pressure-overlap measurements.
Bernoulli samples from
ideal-circuit probabilities drive a four-parameter hybrid inversion. A two-qubit
VQLS circuit then represents the normalized update direction;
normal-system assembly, line search, and model refresh remain classical.
We verify the discrete Born, adjoint, and Gauss--Newton normal actions
using finite differences, independently coded tangent and reverse-adjoint
recurrences, autodiff JVP/VJP evaluations, and explicit-Jacobian
comparisons. Smooth periodic refinement confirms
second-order convergence, whereas omitting receiver calibration leaves an
order-one Born error and substantially changes the regularized
Gauss--Newton direction. All ten predeclared finite-shot runs reduce the
initial model error. The resulting finite-dimensional construction
specifies the physical-pressure derivative and selected-output measurements
needed to connect Schr\"odingerised propagation to a local FWI update.
\end{abstract}



\section{Introduction}

Full-waveform inversion (FWI) repeatedly evaluates forward pressure data and
applies Born, adjoint, and Gauss--Newton actions to recover wavespeed models
from measured waveforms. We construct a pressure-consistent operator-and-readout
formulation of these actions from Schr\"odingerised acoustic propagation.
The construction turns on the energy variables
\(\pi=c^{-1}\partial_tu\) and \(\bm q=\nabla u\): they yield an
auxiliary-space Hamiltonian, but physical pressure is the
coefficient-dependent observable \(p=c\pi\). Consequently, its directional
derivative is
\[
  \D(c\pi)[c_0](\delta c)
  =c_0\,\delta\pi+\delta c\,\pi_0 ,
\]
where the first term is the propagated sensitivity and the second is the
direct receiver-calibration contribution. Proposition~\ref{prop:sisc-receiver-calibration}
shows that both enter the Born map and therefore determine its adjoint and
Gauss--Newton normal action.

FWI is used in seismic imaging, nondestructive evaluation, and ultrasound
computed tomography \cite{Virieux2009,Wang2015WISE,Lucka2022UST}. Its
computational cost is dominated by repeated wave-equation and sensitivity
solves: gradient methods require forward and adjoint propagation, while
Gauss--Newton and related Hessian-informed methods also apply \(Jv\),
\(J^\top r\), and \(J^\top Jv\)
\cite{Bunks1995,Plessix2006,Pratt1998,Metivier2013}. These costs motivate
quantum representations of wave evolution. Quantum algorithms provide
primitives for sparse linear systems, Hamiltonian simulation, and linear
differential equations under explicit access, conditioning, preparation,
and measurement assumptions
\cite{Harrow2009,Childs2017,Berry2015,LowChuang2019,
Berry2017LDE,ChildsLiuOstrander2021}. A Hamiltonian propagation primitive,
however, does not by itself define the differentiated physical receiver
row, its adjoint under the chosen inner products, or the measured quantities
needed for a local update.

Schr\"odingerisation embeds broad classes of linear equations in
Schr\"odinger-type dynamics by adding an auxiliary variable
\cite{JinLiuYu2023PRA,JinLiuYu2024PRL}, and circuit constructions have been
studied for representative partial differential equations and wave systems
\cite{HuJinLiuZhang2024,JinZhang2025}. Related approaches embed parameterized
circuits in a physics-informed FWI architecture
\cite{NguyenVashisthTura2026} or coherently block-encode a PDE-constrained
objective for quantum optimization \cite{SatoKatoYanoItoYamamoto2025}. The
present work instead addresses the coefficient-dependent physical-pressure
derivative and constructs compatible Born, adjoint, and Gauss--Newton actions
together with the selected pressure and Born observables required by a
classically orchestrated local update.

We combine Duhamel differentiation with the receiver-row derivative to obtain
the complete pressure Born map. We prove a conditional consistency estimate
and a periodic second-order finite-difference specialization, and we formulate
a resource model that separately accounts for access, state preparation,
input loading, normalization, quadrature, and selected-output measurement.
The numerical evidence proceeds from independent discrete-action checks,
through identical-matrix statevector agreement, to a compiled finite-shot
four-parameter hybrid update. Sections~2--4 develop the pressure map,
Schr\"odingerised actions, analysis, and resource model; Section~5 gives the
compiled realization; and Sections~6--8 present the evidence, discussion, and
conclusion.

\section{Acoustic Born Primitives and Calibrated Pressure}
\label{sec:sisc-acoustic-born-pressure}

\subsection{Born, adjoint, and Gauss--Newton actions}

For a constant-density acoustic model, let \(u_s(c)\) solve
\begin{equation}
  c(x)^{-2}\partial_{tt}u_s-\Delta u_s=f_s,\qquad
  d_{s,j}(t;c)=\ell_j(\partial_t u_s(t;c)).
\end{equation}
The scalar \(u_s\) is a velocity potential, so the measured pressure trace is represented by \(p_{\rm phys}=\partial_t u_s\), up to a constant-density scaling and sign convention. Stacking sources, receivers, and time samples gives the pressure-data map \(F(c)\). With observed data \(d^{\rm obs}\), the local least-squares objective is
\begin{equation}
  \Phi(c)=\frac12\|F(c)-d^{\rm obs}\|_{W_d}^2,
  \qquad
  \|r\|_{W_d}^2=\langle W_dr,r\rangle .
  \label{eq:sisc-fwi-loss}
\end{equation}
The numerical diagnostics below use \(W_d=I\), but writing the weighted form makes clear where receiver-time quadrature, source weights or noise covariance would enter.

Using the first-order variables \(w=(\pi,\bm q)^\top=(c^{-1}\partial_t u,\nabla u)^\top\), physical pressure is \(\partial_t u=c\pi\). For fixed physical forcing \(f_s\), the corresponding first-order source is coefficient dependent:
\begin{equation}
  \partial_t w=A(c)w+b_s(c,t),\qquad
  b_s(c,t)=\begin{bmatrix}c f_s(t)\\0\end{bmatrix}.
  \label{eq:sisc-first-order-source}
\end{equation}
Its tangent equation therefore contains the source term
\(\D b_s[c_0](\delta c)=(\delta c\,f_s,0)^\top\);
Section~\ref{sec:sisc-schrodingerised-stack} states the fixed-source
convention used by the finite-dimensional diagnostics. Under either
convention, a perturbation \(\delta c\) produces the receiver derivative
\begin{equation}
  (L\delta c)_{s,j,t}=\ell_j\!\left(\delta c\,\pi_s^0+c_0\,\delta\pi_s\right),
  \label{eq:sisc-pressure-born}
\end{equation}
where \(\delta\pi_s\) is the tangent field for the chosen source
convention. The first term in \eqref{eq:sisc-pressure-born} is the
receiver-calibration derivative; the tests below show an \(O(1)\)
finite-difference error when it is omitted.

Equivalently, at the second-order potential level, if \(u_s=u_{0,s}+\delta u_s+O(\|\delta c\|^2)\), then
\begin{equation}
  c_0^{-2}\partial_t^2\delta u_s-\Delta\delta u_s
  =
  \frac{2\,\delta c}{c_0^3}\partial_t^2 u_{0,s},
  \label{eq:sisc-born-pde}
\end{equation}
and \((L\delta c)_{s,j,t}=\ell_j(\partial_t\delta u_s(t))\) after applying the pressure receivers. Formula \eqref{eq:sisc-pressure-born} is the corresponding receiver linearization written in the first-order energy variables needed by the Schr\"odingerised construction. Equation~\eqref{eq:sisc-born-pde} uses fixed physical forcing; a first-order realization of that map must also retain the source derivative in \eqref{eq:sisc-first-order-source}.

\paragraph{Discrete adjoints}
Let the data and model inner products be
\[
  \langle a,b\rangle_d=a^\top W_db,
  \qquad
  \langle u,v\rangle_m=u^\top M_mv,
\]
where \(W_d\) and \(M_m\) are symmetric positive definite. The adjoint of \(L\) under these inner products is
\begin{equation}
  L^*=M_m^{-1}L^\top W_d .
  \label{eq:sisc-weighted-adjoint}
\end{equation}
Therefore, with \(r=F(c_0)-d^{\rm obs}\),
\begin{equation}
  g(c_0)=L^*r,
  \qquad
  H_{\rm GN}(c_0)v=L^*Lv.
\end{equation}
Throughout, ``Hessian action'' denotes this Gauss--Newton normal action
\(L^*Lv\).
Except for the dedicated weighted-adjoint check in Table~\ref{tab:sisc-matched-pressure-baseline}, the reported finite-dimensional identities use \(W_d=M_m=I\) after pressure scaling, so \(L^*=L^\top\). All derivatives are Fr\'echet derivatives with respect to wavespeed \(c\). Slowness and squared-slowness parameterizations produce different Jacobian and adjoint factors through the chain rule.

\begin{proposition}[Receiver-calibration derivative]
\label{prop:sisc-receiver-calibration}
Let \(X\) be the acoustic state space, let \(\ell_j\in X^*\) be independent of \(c\), and suppose \(w(c)=(\pi(c),\bm q(c))\) is Fr\'echet differentiable at \(c_0\). Define \(\delta w=\D w[c_0](\delta c)=(\delta\pi,\delta\bm q)\). For the physical receiver row \(\Gamma_j(c)w=\ell_j(c\pi)\),
\begin{equation}
  \D[\Gamma_j(c)w(c)]_{c_0}(\delta c)
  =\ell_j(c_0\delta\pi)+\ell_j(\delta c\,\pi_0).
  \label{eq:sisc-calibration-lemma}
\end{equation}
\end{proposition}

\begin{proof}
The identity follows from the product rule applied to
\(\Gamma_j(c)w(c)=\ell_j(c\pi(c))\).
\end{proof}

Thus the physical-pressure derivative contains both the propagated tangent
term and the receiver-row derivative. Omitting
\(\ell_j(\delta c\,\pi_0)\) differentiates the frozen observable
\(c_0\pi(c)\) rather than the physical pressure \(c\pi(c)\).
Section~\ref{sec:numerical-evidence} quantifies the resulting Born and
update-level errors.

\paragraph{Fixed-interface derivative convention}
The continuous derivative keeps the material partition fixed; moving
interfaces require additional shape-derivative terms. High-contrast
arrays are therefore treated as finite-dimensional pixel models, and
their checks differentiate the implemented map with respect to grid
values \(c_h\).

\section{Schr\"odingerised Pressure-Observable Born Representation}
\label{sec:sisc-schrodingerised-stack}

\subsection{Auxiliary-space Hamiltonian}

The energy scaling serves two roles. The variable
\(\pi=c^{-1}\partial_tu\) transfers the coefficient-weighted acoustic
energy to standard \(L^2\) coordinates and yields the skew-adjoint block
form below. It also makes physical pressure the coefficient-dependent
row \(c\pi\). Thus the
receiver-calibration term is the receiver-side consequence of the same
scaling that enables the Hamiltonian construction.

Ignoring source injection, the first-order acoustic variables satisfy
\begin{equation}
  \partial_t w=A(c)w,\qquad
  A(c)=
  \begin{bmatrix}
  0 & c\nabla\cdot\\
  \nabla(c\,\cdot) & 0
  \end{bmatrix}.
  \label{eq:sisc-acoustic-A}
\end{equation}
For periodic or energy-conserving boundary conditions this operator is skew-adjoint in the energy variables. Absorbing layers, damping and stabilized closures add non-skew components, so we write \(A(c)=A_{\rm sk}(c)+A_{\rm h}(c)\), with \(A_{\rm sk}^*=-A_{\rm sk}\) and \(A_{\rm h}^*=A_{\rm h}\). Schr\"odingerisation introduces an auxiliary coordinate \(p_a\) and the Hermitian Hamiltonian
\begin{equation}
  \calH_{\rm ac}(c)=\ii A_{\rm sk}(c)\otimes I_{p_a}
  -\ii A_{\rm h}(c)\otimes D_{p_a},\qquad
  U_c(t)=e^{-\ii \calH_{\rm ac}(c)t}.
  \label{eq:sisc-hamiltonian}
\end{equation}
In the finite-dimensional implementation, \(A_{{\rm sk},h}\) is the centered first-order acoustic coupling on the physical grid, \(A_{{\rm h},h}\) contains the coefficient-dependent symmetric and damping terms, and \(D_{p_a,h}\) is the centered derivative on the auxiliary grid. Thus
\begin{equation}
  H_h(c_h)=\ii A_{{\rm sk},h}(c_h)\otimes I
  -\ii A_{{\rm h},h}(c_h)\otimes D_{p_a,h},
  \label{eq:sisc-fd-hamiltonian}
\end{equation}
with the corresponding coefficient derivative
\begin{equation}
  \D H_h[c_h](\delta c_h)=
  \ii\,\D A_{{\rm sk},h}[c_h](\delta c_h)\otimes I
  -\ii\,\D A_{{\rm h},h}[c_h](\delta c_h)\otimes D_{p_a,h}.
  \label{eq:sisc-fd-hamiltonian-derivative}
\end{equation}
For the two-dimensional periodic stencil used in the calibrated matrix and statevector diagnostics, let \(D_x,D_z\in\mathbb R^{N\times N}\), \(N=n_xn_z\), be centered difference matrices with \(D_x^\top=-D_x\) and \(D_z^\top=-D_z\). With \(C_h=\operatorname{diag}(c_h)\), the skew acoustic block is
\begin{equation}
  A_{{\rm sk},h}(c_h)
  =
  \begin{bmatrix}
  0 & C_hD_x & C_hD_z\\
  D_xC_h & 0 & 0\\
  D_zC_h & 0 & 0
  \end{bmatrix},
  \qquad
  A_{{\rm sk},h}(c_h)^\top=-A_{{\rm sk},h}(c_h).
  \label{eq:sisc-Ask-block}
\end{equation}
The damped implementation sets \(A_{{\rm h},h}=-\Sigma_h\), where \(\Sigma_h\succeq0\) is diagonal and repeated over the acoustic components, so
\begin{equation}
  H_h(c_h)=\ii A_{{\rm sk},h}(c_h)\otimes I_{p_a}
  +\ii\Sigma_h\otimes D_{p_a,h}.
  \label{eq:sisc-Hh-implemented}
\end{equation}
This is Hermitian because \(A_{{\rm sk},h}^\top=-A_{{\rm sk},h}\),
\(\Sigma_h^\top=\Sigma_h\), and
\(D_{p_a,h}^\top=-D_{p_a,h}\). The continuous statements assume that
\(\ii A_{\rm sk}\otimes I\), \(-\ii A_{\rm h}\otimes D_{p_a}\), and
their sum are self-adjoint on a common dense domain; relative
boundedness gives one sufficient condition. The numerical results use
the finite-dimensional Hermitian matrix \(H_h\), for which the displayed
identity is exact. After discretizing the physical variables but before
discretizing \(p_a\), the warped equation is
\(\partial_t\psi_h+A_{{\rm h},h}\partial_{p_a}\psi_h
=A_{{\rm sk},h}\psi_h\) \cite{JinLiuYu2023PRA}. Since
\(A_{{\rm h},h}\) is Hermitian, a standard energy estimate shows that
\(\lambda_{+,h}:=\max\{0,\lambda_{\max}(A_{{\rm h},h})\}\) bounds the
maximal right-going auxiliary characteristic speed. Consequently, on the
untruncated auxiliary line, any recovery functional supported in
\(p_a\geq p_\star>\lambda_{+,h}T\) is unaffected, up to time \(T\), by
how the initial warp is continued into \(p_a<0\). For the periodic auxiliary
grid used in the numerical specialization, finite-domain and periodization
effects are included in the assumed \(C e^{-\alpha L_p}\) auxiliary
truncation bound of Proposition~\ref{prop:sisc-periodic}. In the present
specialization, \(A_{{\rm h},h}=-\Sigma_h\preceq0\) gives
\(\lambda_{+,h}=0\), so the recovery support need not be shifted rightward
by an amount growing with \(T\). Thus \(p_\star\) can be fixed independently
of \(T\), and the recovery factor \(e^{p_\star}\) introduces no
\(T\)-dependent exponential amplitude cost. Proposition~\ref{prop:sisc-recovery-transfer}
controls the remaining background and derivative recovery-map discrepancies
abstractly. A time-domain PML must first be assembled as an augmented
generator, including its split-field or auxiliary variables, before applying
this decomposition. The local derivative freezes its damping profile
\(\sigma\); coefficient-dependent damping would add
\(\D\sigma[c_0](\delta c)\) rows.

The derivative of the acoustic block is explicit:
\begin{equation}
  \D A_{{\rm sk},h}[c_h](\delta c_h)
  =
  \begin{bmatrix}
  0 & \Delta C_hD_x & \Delta C_hD_z\\
  D_x\Delta C_h & 0 & 0\\
  D_z\Delta C_h & 0 & 0
  \end{bmatrix},
  \qquad
  \Delta C_h=\operatorname{diag}(\delta c_h).
  \label{eq:sisc-dAsk-block}
\end{equation}
Equivalently, for \(w=(\pi,q_x,q_z)^\top\), the derivative action can be viewed as the linear map
\begin{equation}
  \delta c_h
  \mapsto
  \begin{bmatrix}
  \operatorname{diag}(D_xq_x+D_zq_z)\\
  D_x\operatorname{diag}(\pi)\\
  D_z\operatorname{diag}(\pi)
  \end{bmatrix}
  \delta c_h .
  \label{eq:sisc-dAsk-action}
\end{equation}
This block-level definition is the finite-dimensional object tested below.

For a source state \(|s\rangle\), the calibrated pressure observable is
\begin{equation}
  d^{\rm Sch}_{s,j}(t;c)=\Gamma_j(c)\calR_{p_a}U_c(t)|s\rangle,\qquad
  \Gamma_j(c)w=\ell_j(c\pi).
  \label{eq:sisc-q-data}
\end{equation}
Here \(\calR_{p_a}\) denotes the recovery or auxiliary averaging operation
used to map the extended state back to the physical acoustic variables.
The tests and Proposition~\ref{prop:sisc-periodic} adopt a centered
periodic auxiliary derivative and a compact receiver profile as the
finite-dimensional specialization.

\paragraph{Source convention}
The finite-dimensional diagnostics represent each source history by a fixed collection of normalized first-order pressure states \(v_{s,k}\), independent of \(c\), and differentiate
\begin{equation}
  c\longmapsto
  \Gamma_j(c)\calR_{p_a}\sum_k U_c(t-t_k)v_{s,k}.
  \label{eq:sisc-fixed-source-state-map}
\end{equation}
Their Duhamel derivative therefore contains the Hamiltonian and receiver
derivatives. Differentiating the fixed-physical-forcing map in
\eqref{eq:sisc-first-order-source} additionally contributes the embedded
term \((\delta c\,f_s,0)^\top\). Supplementary finite differences verify
both conventions and show an \(O(1)\) error when this source term is excluded
from the fixed-physical-forcing map.

\begin{proposition}[Transfer of auxiliary-recovery error to pressure Born rows]
\label{prop:sisc-recovery-transfer}
Let \(E_s:\mathcal X\to\mathcal H_{\rm ext}\) be a coefficient-independent source embedding and set \(v_s=E_sw_s\). Let \(S(c;t):\mathcal X\to\mathcal X\) be the physical first-order acoustic solution operator under the same fixed source-state convention. Assume, uniformly for \(0\le t\le T\),
\begin{align}
  \|\calR_{p_a}U_c(t)E_s-S(c;t)\|_{\mathcal L(\mathcal X)}
  &\le \varepsilon_R, \nonumber\\
  \big\|\D[\calR_{p_a}U_c(t)E_s]_{c_0}
  -\D[S(c;t)]_{c_0}\big\|_{\mathcal L(\mathcal C_r,\mathcal L(\mathcal X))}
  &\le \varepsilon_R.
  \label{eq:sisc-recovery-assumption}
\end{align}
Assume also that \(\|\Gamma_j(c_0)\|_{\mathcal X^*}\le C_\Gamma\) and
\[
  \|\D\Gamma_j[c_0](\delta c)\|_{\mathcal X^*}
  \le C_{\D\Gamma}\|\delta c\|_{\mathcal C_r}.
\]
Then
\begin{align}
 &\left|
 \D\!\left[\Gamma_j(c)\calR_{p_a}U_c(t)E_sw_s\right]_{c_0}(\delta c)
 -\D\!\left[\Gamma_j(c)S(c;t)w_s\right]_{c_0}(\delta c)
 \right| \nonumber\\
 &\qquad\le
 (C_\Gamma+C_{\D\Gamma})\varepsilon_R
 \|\delta c\|_{\mathcal C_r}\|w_s\|_{\mathcal X}.
\end{align}
A stacked data-norm bound follows after including the stated source, receiver, and time weights.
\end{proposition}

\begin{proof}
Set
\(E(c,t)=\calR_{p_a}U_c(t)E_s-S(c;t)\).
The product rule gives the exact decomposition
\[
 \Gamma_j(c_0)\D E[c_0,t](\delta c)w_s
 +\D\Gamma_j[c_0](\delta c)E(c_0,t)w_s .
\]
The two terms are bounded by \(C_\Gamma\varepsilon_R\) and
\(C_{\D\Gamma}\varepsilon_R\), respectively, times
\(\|\delta c\|_{\mathcal C_r}\|w_s\|_{\mathcal X}\). Weighted
summation gives the stacked estimate.
\end{proof}

\subsection{Duhamel Born map with calibrated pressure rows}

Under the coefficient-independent source-state convention \eqref{eq:sisc-fixed-source-state-map}, the first variation follows from Duhamel's formula,
\begin{equation}
  \delta\psi_s(T)=
  -\ii\int_0^T U_{c_0}(T-\tau)\D\calH_{\rm ac}[c_0](\delta c)
  U_{c_0}(\tau)|s\rangle\,\dd\tau ,
\end{equation}
and the calibrated Born row is
\begin{equation}
  (L_{\rm Sch}^{\rm ac}\delta c)_{s,j,T}
  =\Gamma_j(c_0)\calR_{p_a}\delta\psi_s(T)
  +\D\Gamma_j[c_0](\delta c)\calR_{p_a}U_{c_0}(T)|s\rangle .
  \label{eq:sisc-sch-born}
\end{equation}
The second term is the pressure receiver derivative in the Schr\"odingerised representation. The adjoint and Hessian actions are defined with the Euclidean inner product after time, source and receiver discretization:
\begin{equation}
  \nabla\Phi_{\rm Sch}(c_0)=L_{\rm Sch}^{\rm ac,*}r,\qquad
  H_{{\rm GN},{\rm Sch}}(c_0)v
  =L_{\rm Sch}^{\rm ac,*}L_{\rm Sch}^{\rm ac}v .
\end{equation}

The formulation uses operator actions throughout. A forward action evaluates the Duhamel integral for a perturbation \(v\); an adjoint action reverses the same time, source, and receiver contractions against a residual \(r\); a Gauss--Newton action composes the two. Hereafter \(J\in\mathbb R^{M\times N}\) denotes the Euclidean coordinate representation of the discrete calibrated Born map \(L_{{\rm Sch},h}^{\rm ac}\), and \(J^\top\) is its Euclidean adjoint. The local inverse calculation uses the three matrix-free actions \(Jv\), \(J^\top r\), and \(J^\top Jv\).

\subsection{Adjoint convention and Hessian-action interface}

Let \(d\in\mathbb R^M\) collect all source, receiver and time samples after the pressure scaling has been restored. For the least-squares objective
\begin{equation}
  \Phi(c)=\frac12\|F(c)-d^{\rm obs}\|_2^2,
\end{equation}
the local residual is \(r=F(c_0)-d^{\rm obs}\), and the local quadratic model is
\begin{equation}
  \Phi(c_0+v)\approx \Phi(c_0)+\langle L v,r\rangle
  +\frac12\langle v,L^*Lv\rangle .
\end{equation}
Replacing \(L\) by \(L_{\rm Sch}^{\rm ac}\) preserves the local quadratic
FWI model only when the corresponding adjoint and normal-action identities
also hold. The tests therefore check finite-difference Born consistency,
adjoint identities
\begin{equation}
  \langle L_{\rm Sch}^{\rm ac}v,r\rangle_{\mathbb R^M}
  =\langle v,L_{\rm Sch}^{\rm ac,*}r\rangle_{\mathbb R^N},
\end{equation}
and Hessian symmetry
\begin{equation}
  \langle u,L_{\rm Sch}^{\rm ac,*}L_{\rm Sch}^{\rm ac}v\rangle
  =\langle L_{\rm Sch}^{\rm ac}u,L_{\rm Sch}^{\rm ac}v\rangle .
\end{equation}
These identities are the mathematical reason for testing the Born map,
adjoint map, and Hessian action as one coupled operator chain.

\subsection{State, receiver and measurement preparation}

Here a \emph{selected-output measurement} denotes a prescribed subset or
linear sketch of receiver--time pressure functionals used in a local
update. The adjective ``selected'' describes this prescribed output target;
the LCU \textsc{Select} unitary, selector-zero block, and postselection are
separate concepts.

\paragraph{State and receiver preparation}
In the finite-dimensional diagnostics, source states are normalized pressure source profiles tensored with the auxiliary profile. Receiver states are frozen at the background \(c_0\) and encode the row \(c_0(x_j)e_{\pi,x_j}\otimes\chi_{\rm rec}\), with the physical pressure scale recorded separately. The assembled-matrix statevector experiments apply the corresponding finite-dimensional states and operators exactly. Section~\ref{sec:sisc-gate-level} additionally compiles structured small-instance source and receiver states, product-formula evolution, a derivative LCU block, and pressure-overlap measurements. The sparse-oracle construction used for asymptotic resource statements remains specified separately below.

\paragraph{Measurement model}
For a real receiver functional \(a=\operatorname{Re}\langle \eta|\psi\rangle\), a Hadamard-test sampling model returns \(X_\ell\in\{-1,+1\}\) with \(\mathbb E X_\ell=a\); the estimator is \(\hat a=N^{-1}\sum_{\ell=1}^N X_\ell\), followed by the stored pressure normalization. Amplitude estimation changes the scalar-tolerance dependence, while entrywise readout of a full gather still carries the factor \(M\).

For simple Hadamard sampling,
\begin{equation}
  \mathbb E\hat a=a,\qquad
  \operatorname{Var}(\hat a)=\frac{1-a^2}{N}\le \frac1N .
  \label{eq:sisc-hadamard-variance}
\end{equation}
If a pressure datum is \(s_{j,h}\hat a\), the absolute standard deviation inherits the same scale factor \(s_{j,h}/\sqrt N\). Thus a full gather of \(M\) independent scalar estimates accumulates measurement uncertainty in the data norm even when the Hamiltonian simulation is exact. A sketched observable changes the target to \(Sd\), or to scalar functionals generated by rows of \(S\). For a local least-squares update this leads to
\begin{equation}
  \Phi_S(c)=\frac12\|S(F(c)-d^{\rm obs})\|_2^2,\qquad
  \nabla\Phi_S(c_0)=J^\top S^\top S r,\qquad
  H_Sv=J^\top S^\top S Jv .
  \label{eq:sisc-sketched-objective}
\end{equation}
The relevant diagnostic is therefore preservation of the gradient and Hessian-action directions used by the local method. We report finite-shot scalar errors and compressed Hessian-action sketches separately: pressure traces define the physical observable, while compressed Hessian-action information is the plausible readout target.

\section{Consistency, Measurement and Resource Accounting}
\label{sec:sisc-measurement-resource}

The following estimates analyze the finite-dimensional implementation
under explicit stability assumptions.

\subsection{Conditional operator consistency}

The next result separates the errors contributed by state evolution,
receiver evaluation, receiver differentiation, and Duhamel quadrature.
Proposition~\ref{prop:sisc-recovery-transfer} supplies the
auxiliary-recovery contributions included in \(\varepsilon_0\) and
\(\varepsilon_1\).

\begin{proposition}[Conditional pressure-Born error budget]
\label{prop:sisc-consistency}
Let \(\mathcal C\) and \(\mathcal C_h\) be continuous and discrete coefficient spaces with interpolation \(I_h^c:\mathcal C\to\mathcal C_h\). Let \(\mathcal C_r\subset\mathcal C\) be the admissible perturbation space, equipped with \(\|\cdot\|_{\mathcal C_r}\). Let \(X\) and \(X_h\) be continuous and discrete acoustic state spaces with projection \(\Pi_h^X:X\to X_h\). Let \(\mathcal D\) and \(\mathcal D_h\) be the corresponding pressure-data spaces with sampling map \(P_h^{\rm obs}:\mathcal D\to\mathcal D_h\). Assume \(c_0(x)\ge c_{\min}>0\), \(c_0+\theta\delta c\ge c_{\min}/2\) for \(0\le\theta\le1\), and that \(I_h^c\) preserves this positivity along the discrete coefficient path, with fixed material interfaces. All continuous and discrete background and tangent maps in this proposition use the same source convention.

Let \(\psi_0,\delta\psi\) and \(\psi_{0,h},\delta\psi_h\) denote the continuous and exact-in-time discrete background and tangent states. Assume, uniformly for \(0\le T\le T_0\),
\begin{align}
 \|\psi_{0,h}-\Pi_h^X\psi_0\|_{X_h}
 &\le \varepsilon_0, \nonumber\\
 \|\delta\psi_h-\Pi_h^X\delta\psi\|_{X_h}
 &\le \varepsilon_1\|\delta c\|_{\mathcal C_r}.
 \label{eq:sisc-state-defects}
\end{align}
Here \(\varepsilon_0\) and \(\varepsilon_1\) include the stated spatial, Hamiltonian-derivative, source, and auxiliary-recovery errors. Assume also
\begin{equation}
  \|\psi_0(T)\|_X\le C_0,
  \qquad
  \|\delta\psi(T)\|_X
  \le C_1\|\delta c\|_{\mathcal C_r},
  \qquad 0\le T\le T_0,
  \label{eq:sisc-state-stability}
\end{equation}
and
\begin{align}
 \|\Gamma_h(I_h^cc_0)\Pi_h^X-P_h^{\rm obs}\Gamma(c_0)\|
 &\le\varepsilon_\Gamma, \nonumber\\
 \|\D\Gamma_h[I_h^cc_0](I_h^c\delta c)\Pi_h^X
 -P_h^{\rm obs}\D\Gamma[c_0](\delta c)\|
 &\le\varepsilon_{\D\Gamma}\|\delta c\|_{\mathcal C_r},
 \label{eq:sisc-receiver-defects}
\end{align}
with uniformly bounded continuous and discrete receiver rows and
\begin{equation}
  \big\|\D\Gamma_h[I_h^cc_0](I_h^c\delta c)\big\|_{\mathcal L(X_h,\mathcal D_h)}
  \le C_{\D\Gamma,h}\|\delta c\|_{\mathcal C_r}.
  \label{eq:sisc-discrete-receiver-derivative-bound}
\end{equation}
Finally, let \(L_{h,Q}^{\rm ac}\) be the \(Q\)-node Duhamel approximation of the exact-in-time discrete operator \(L_h^{\rm ac}\), and assume
\begin{equation}
 \|(L_{h,Q}^{\rm ac}-L_h^{\rm ac})I_h^c\delta c\|_{\mathcal D_h}
 \le\varepsilon_Q\|\delta c\|_{\mathcal C_r}.
 \label{eq:sisc-quadrature-defect}
\end{equation}
Then
\begin{equation}
 \|L_{h,Q}^{\rm ac}I_h^c\delta c
   -P_h^{\rm obs}L_{\rm Sch}^{\rm ac}\delta c\|_{\mathcal D_h}
 \le C\big(\varepsilon_0+\varepsilon_1+\varepsilon_\Gamma
   +\varepsilon_{\D\Gamma}+\varepsilon_Q\big)
   \|\delta c\|_{\mathcal C_r}.
 \label{eq:sisc-consistency}
\end{equation}
\end{proposition}

\begin{proof}
First separate the quadrature error from the exact-in-time discrete
operator:
\begin{equation}
  L_{h,Q}^{\rm ac}I_h^c\delta c-P_h^{\rm obs}L_{\rm Sch}^{\rm ac}\delta c
  =
  (L_{h,Q}^{\rm ac}-L_h^{\rm ac})I_h^c\delta c
  +(L_h^{\rm ac}I_h^c\delta c-P_h^{\rm obs}L_{\rm Sch}^{\rm ac}\delta c).
  \label{eq:sisc-error-split}
\end{equation}
Insert the projected continuous background and tangent states into the
second term. The resulting state and receiver decomposition is
\begin{align}
 L_h^{\rm ac}I_h^c\delta c-P_h^{\rm obs}L_{\rm Sch}^{\rm ac}\delta c
 &=\Gamma_h(I_h^cc_0)(\delta\psi_h-\Pi_h^X\delta\psi) \nonumber\\
 &\quad+(\Gamma_h(I_h^cc_0)\Pi_h^X
   -P_h^{\rm obs}\Gamma(c_0))\delta\psi \nonumber\\
 &\quad+\D\Gamma_h[I_h^cc_0](I_h^c\delta c)
   (\psi_{0,h}-\Pi_h^X\psi_0) \nonumber\\
 &\quad+(\D\Gamma_h[I_h^cc_0](I_h^c\delta c)\Pi_h^X
   -P_h^{\rm obs}\D\Gamma[c_0](\delta c))\psi_0 .
  \label{eq:sisc-born-error-components}
\end{align}
The four terms in \eqref{eq:sisc-born-error-components} are bounded,
respectively, by constants times
\(\varepsilon_1\), \(\varepsilon_\Gamma C_1\),
\(C_{\D\Gamma,h}\varepsilon_0\), and
\(\varepsilon_{\D\Gamma}C_0\), each multiplied by
\(\|\delta c\|_{\mathcal C_r}\). Adding the quadrature bound
\eqref{eq:sisc-quadrature-defect} proves \eqref{eq:sisc-consistency}.
The constant \(C\) depends on the stated receiver-row and state bounds,
the source/receiver/time weights, and \(T_0\), and is uniform in
\(\delta c\) and the five displayed tolerances.
\end{proof}

The estimate controls the pressure-observable Born map that underlies
the adjoint and Hessian actions. It gives consistency under
the stated spatial discretization, auxiliary recovery, receiver
projection, and quadrature assumptions. The convergence and conditioning
of the nonlinear FWI iteration are governed by the corresponding
inverse-problem geometry.

Equation~\eqref{eq:sisc-born-error-components} identifies the receiver-calibration derivative as a separate consistency term. After two operators are represented on the same finite-dimensional model and data spaces, their normal-action error satisfies
\[
  \|A^\top A-B^\top B\|
  \le(\|A\|+\|B\|)\|A-B\|.
\]
The symmetry identities tested below are exact algebraic identities for the adjoint of the same discrete Born map. Consistency with a projected continuous normal action additionally requires the compatible model and data inner products in \eqref{eq:sisc-weighted-adjoint}.

\subsection{Periodic finite-difference specialization}

\begin{proposition}[Periodic finite-difference specialization]
\label{prop:sisc-periodic}
For this proposition set
\(\mathcal C_r=\mathcal C_3=C_{\rm per}^3(\Omega)\), equipped with the
standard \(C^3\) norm, and fix \(s>d/2\). Measure continuous state errors
in \(H_{\rm per}^s(\Omega)^{d+1}\) and equip the grid states with a stable
discrete \(H^s\) norm. Then
\(H_{\rm per}^s\hookrightarrow C(\overline\Omega)\), so the stated
point-sampling or interpolation receivers are uniformly bounded.

For the periodic finite-difference implementation used in the calibrated
acoustic \(p_a\)-space diagnostics, suppose the spatial and auxiliary
derivatives are centered periodic differences and the damping term is
diagonal and fixed. Assume sufficient periodic smoothness of the
coefficient, source, receiver, auxiliary profile, and recovery map and,
uniformly for \(0\le t\le T_0\),
\[
  \|\psi_0(t)\|_{H_{\rm per}^{s+3}}\le C_0,
  \qquad
  \|\delta\psi(t)\|_{H_{\rm per}^{s+3}}
  \le C_1\|\delta c\|_{\mathcal C_3}.
\]
Assume further that every coefficient--state product to which a
centered first difference is applied in
\(A_{\rm sk}(c_0)\psi_0\) and
\(\D A_{\rm sk}[c_0](\delta c)\psi_0\) belongs to
\(H_{\rm per}^{s+3}\), uniformly on \([0,T_0]\), and that the products
in the derivative term have norms bounded by
\(C_2\|\delta c\|_{\mathcal C_3}\).

Finally, assume that the continuous and semidiscrete background and tangent
evolution families are uniformly stable on \([0,T_0]\) in these norms,
that the source embedding and state projection are second-order
consistent, and that the chosen auxiliary profile and recovery map
satisfy a truncation estimate \(Ce^{-\alpha L_p}\) and a
centered-difference estimate \(C\Delta p^2\) in the norms of
Proposition~\ref{prop:sisc-consistency}. Under these hypotheses, the
physical-grid contribution is second order, and the total estimate inherits
the assumed auxiliary rates:
\begin{equation}
  \|L_{{\rm per},h,Q}^{\rm ac}I_h^c\delta c-P_h^{\rm obs}L_{\rm Sch}^{\rm ac}\delta c\|
  \le C_{T_0}\big(h_x^2+h_z^2+\Delta p^2+e^{-\alpha L_p}
  +\varepsilon_Q\big)\|\delta c\|_{\mathcal C_3}.
  \label{eq:sisc-periodic}
\end{equation}
The discrete pressure row includes
\begin{align*}
  R_{j,h}(c_h)V_h
  &=c_h(x_j)\sum_{\rho=1}^{N_p}\chi_{{\rm rec},\rho}\pi_h(x_j,\rho),\\
  \D R_{j,h}[c_h](\delta c_h)V_h
  &=\delta c_h(x_j)\sum_{\rho=1}^{N_p}\chi_{{\rm rec},\rho}\pi_h(x_j,\rho).
\end{align*}
Thus the receiver-calibration derivative is part of the implemented Born
operator.
\end{proposition}

\begin{proof}
The standard centered-difference estimate
\[
  \|D_hI_hv-I_hDv\|_{H_h^s}
  \le Ch^2\|v\|_{H^{s+3}}
\]
gives a second-order local defect for the background equation. In the
tangent equation, the stated product regularity controls both
\(A(c_0)\delta\psi\) and
\(\D A[c_0](\delta c)\psi_0\) at the same order, multiplied by
\(\|\delta c\|_{\mathcal C_3}\). Source/projection consistency, variation of
constants, and uniform stability then give the
\(h_x^2+h_z^2\) contribution; recovery gives
\(\Delta p^2+e^{-\alpha L_p}\). The displayed receiver derivative and
Proposition~\ref{prop:sisc-consistency} yield
\eqref{eq:sisc-periodic}, with \(\varepsilon_Q\) retaining quadrature
error. This estimate concerns the smooth periodic setting.

\end{proof}

The spatial \(O(h_x^2+h_z^2)\) rate follows from centered-difference
consistency. The numerical auxiliary row below reports the trend implied
by the assumed recovery rates.

\begin{corollary}[Computed pressure-Born error]
\label{cor:sisc-stack-error}
Let \(\widehat L_{h,Q}^{\rm ac}\) be a computed approximation of \(L_{h,Q}^{\rm ac}\). Assume that, for every \(\delta c\), the simulation, preparation and loading terms are operator errors, while \(\varepsilon_{\rm meas}^{(B)}\) is an additive measurement error holding with the stated success probability:
\begin{equation}
 \|(\widehat L_{h,Q}^{\rm ac}-L_{h,Q}^{\rm ac})I_h^c\delta c\|_{\mathcal D_h}
 \le(\varepsilon_{\rm sim}+\varepsilon_{\rm prep}
   +\varepsilon_{\rm load})\|\delta c\|_{\mathcal C_r}
   +\varepsilon_{\rm meas}^{(B)}.
 \label{eq:sisc-computed-operator-error}
\end{equation}
Then, under Proposition~\ref{prop:sisc-consistency},
\begin{align}
  \|\widehat L_{h,Q}^{\rm ac}I_h^c\delta c
  -P_h^{\rm obs}L_{\rm Sch}^{\rm ac}\delta c\|_{\mathcal D_h}
  &\le C\big(\varepsilon_0+\varepsilon_1+\varepsilon_\Gamma
  +\varepsilon_{\D\Gamma}+\varepsilon_Q \nonumber\\
  &\qquad\quad
  +\varepsilon_{\rm sim}+\varepsilon_{\rm prep}
  +\varepsilon_{\rm load}\big)
  \|\delta c\|_{\mathcal C_r}
  +\varepsilon_{\rm meas}^{(B)}.
  \label{eq:sisc-stack-error}
\end{align}
\end{corollary}

Corollary~\ref{cor:sisc-stack-error} separates finite-difference
consistency from errors due to source and receiver preparation,
perturbation or residual loading, Hamiltonian simulation, and
measurement. Table~\ref{tab:sisc-resource} records the corresponding
query, gate, memory, and preparation costs. For entrywise full-gather
readout, \(\varepsilon_{\rm meas}^{(B)}\) includes accumulation over
the measured receiver--time scalars; for a sketched observable it is
the error after applying the sketch.

\subsection{Frozen receiver-state calibration}

The statevector diagnostics use normalized states, whereas the receiver row in \eqref{eq:sisc-q-data} is a physical pressure functional. Let \(|\eta_j(c_0)\rangle\) be the normalized receiver state proportional to the frozen row \(c_0(x_j)e_{\pi,x_j}\otimes\chi_{\rm rec}\), and let \(s_{j,h}(c_0)\) be its stored norm and pressure scale. Then
\begin{equation}
  \Gamma_{j,h}(c_0)V_h
  =s_{j,h}(c_0)\langle \eta_j(c_0),V_h\rangle .
  \label{eq:sisc-receiver-normalization}
\end{equation}
The implementation forms the Born row by freezing the receiver state at \(c_0\), tracking the physical scale classically, and adding the explicit derivative \(\D\Gamma_j[c_0](\delta c)\) in \eqref{eq:sisc-sch-born}. Writing \(\psi_0(T)=U_{c_0}(T)v_s\), the frozen-state and physical-pressure derivatives are
\begin{align}
  \delta d^{\rm frozen}_{j,h}
  &=
  s_{j,h}(c_0)\langle\eta_j(c_0),\calR_{p_a}\delta\psi_s(T)\rangle ,
  \label{eq:sisc-frozen-state-derivative}\\
  \delta d^{\rm press}_{j,h}
  &=
  s_{j,h}(c_0)\langle\eta_j(c_0),\calR_{p_a}\delta\psi_s(T)\rangle
  +\D\Gamma_{j,h}[c_0](\delta c)\calR_{p_a}\psi_0(T).
  \label{eq:sisc-pressure-alignment}
\end{align}
Equation \eqref{eq:sisc-frozen-state-derivative} is what a frozen normalized measurement state supplies. Equation \eqref{eq:sisc-pressure-alignment} is the physical pressure Born row. The gap between them is exactly the receiver-calibration derivative. If the measurement state were updated as a differentiable function of \(c\), additional normalization and state-derivative terms would appear; the implementation avoids this ambiguity by freezing the measurement state at the current background and adding the unnormalized pressure-row derivative explicitly.

This is the finite-dimensional meaning of the calibrated pressure observable. Raw amplitudes and uncalibrated \(\pi\)-rows correspond to different observables. Operationally, a hardware-level implementation would have to prepare \(|\eta_j(c_0)\rangle\), track or estimate \(s_{j,h}(c_0)\), and refresh these objects when the nonlinear background changes. A receiver supported on \(m_j\) physical grid points and \(N_p\) auxiliary nodes has \(O(m_jN_p)\) nonzero amplitudes before normalization. The statevector diagnostics evaluate the calibrated row directly on stored arrays; the resource table records receiver preparation and normalization separately from the Hamiltonian oracle.

\subsection{Oracle costs and selected-output readout}

The sparse finite-difference stencil gives row/value oracles for \(H_h(c_h)\) and \(\D H_h[c_h](\delta c_h)\), with block-encoding normalization scales \(\lambda_H\) and \(\lambda_{\D H}\). The resource model specifies these oracles and the inverse-problem inputs separately. Explicit block encodings of structured sparse matrices require concrete circuit constructions and access models \cite{CampsLinVanBeeumenYang2022,DanzStollenwerkCiani2026}. Index a \(p_a\)-space state by \((\alpha,i,\rho)\), where \(\alpha\in\{\pi,q_x,q_z\}\), \(i\) is a physical grid point and \(\rho\) is an auxiliary node. For the periodic stencil in \eqref{eq:sisc-Ask-block}, the acoustic rows contain
\begin{align}
  H_{(\pi,i,\rho),(q_x,i\pm e_x,\rho)}
  &=\pm \ii\,c_i/(2h_x),&
  H_{(\pi,i,\rho),(q_z,i\pm e_z,\rho)}
  &=\pm \ii\,c_i/(2h_z), \nonumber\\
  H_{(q_x,i,\rho),(\pi,i\pm e_x,\rho)}
  &=\pm \ii\,c_{i\pm e_x}/(2h_x),&
  H_{(q_z,i,\rho),(\pi,i\pm e_z,\rho)}
  &=\pm \ii\,c_{i\pm e_z}/(2h_z).
  \label{eq:sisc-oracle-stencil}
\end{align}
The \(p_a\)-coupling from a diagonal damping block connects \(\rho\) to
\(\rho\pm1\). A periodic location oracle therefore performs modular
addition in \(x\), \(z\), and \(p_a\), while the value oracle queries
\(c_i\) and attaches the finite-difference factor and phase \(\ii\).
Set \(c_{\max}=\|c_h\|_\infty\),
\(h_{\min}=\min\{h_x,h_z\}\), and
\(\sigma_{\max}=\|\Sigma_h\|_\infty\). With \(s_{\rm sk}=4\), diagonal
damping sparsity \(s_{\rm h}=1\), and central auxiliary derivative sparsity
\(s_p=2\),
\begin{equation}
  s_H\le s_{\rm sk}+s_{\rm h}s_p,\qquad
  \lambda_H=
  O\!\left(s_{\rm sk}\frac{c_{\max}}{h_{\min}}
  +s_{\rm h}s_p\frac{\sigma_{\max}}{\Delta p}\right).
  \label{eq:sisc-lambda-H}
\end{equation}
The derivative oracle uses the same location map, with every queried coefficient \(c_i\) in the acoustic part replaced by the input perturbation \(\delta c_i\). For fixed damping,
\begin{equation}
  \lambda_{\D H}=
  O\!\left(s_{\rm sk}\frac{\|\delta c\|_\infty}{h_{\min}}\right).
  \label{eq:sisc-lambda-dH}
\end{equation}
If damping, density or transducer terms are made model dependent, their derivative rows enter \(\lambda_{\D H}\) as additional terms.

Let \(\mathcal O_c\) and \(\mathcal O_v\) denote value oracles for \(c_h\) and a perturbation \(v_h\). Let \(U_s\), \(U_\eta\), and \(U_r\) prepare source, frozen receiver, and residual states, including their normalization factors. For source and receiver profiles represented on known supports of sizes \(k_s\) and \(k_\eta\), a direct construction enumerates the nonzero grid--auxiliary entries, applies controlled rotations, and uncomputes the support index. Its gate counts are linear in the support sizes up to index arithmetic. A discretized Gaussian source represented on the full grid has \(k_s\) equal to that full support. Preparing a generic dense residual or perturbation from a classical array similarly incurs work proportional to its input length unless further structure is available. A compressed observable additionally uses an oracle or preparation map \(U_S\) for its sketch rows. We cost value-oracle access, amplitude-state preparation, and input conversion separately; value-oracle access carries the stated sparse structure. Classical matrix-free tangent and adjoint actions likewise receive \(v\) and \(r\) as explicit inputs.

For a fixed background \(c_h=c_{0,h}\), quadrature nodes \(\tau_\ell\),
and weights \(w_\ell\), the discrete calibrated Born action has the form
\begin{align}
  L_{{\rm Sch},h}^{\rm ac}\delta c_h
  &=
  -\ii\,\Gamma_h(c_h)\calR_{p_a}
  \sum_{\ell=1}^{Q}
  w_\ell\,
  U_h(T-\tau_\ell)
  \D H_h[c_h](\delta c_h)
  U_h(\tau_\ell)v_s \nonumber\\
  &\quad
   +\D\Gamma_h[c_h](\delta c_h)
    \calR_{p_a}U_h(T)v_s ,
  \label{eq:sisc-born-lcu}
\end{align}
where the last term is the direct receiver-calibration term. A direct
\(Q\)-node LCU construction uses \(O(Q)\) controlled
Hamiltonian-evolution segments and \(O(Q)\) calls to a block encoding of
\(\D H_h[c_h](\delta c_h)\). Define the LCU coefficient normalization
\begin{equation}
  \alpha_Q=\lambda_{\D H}\sum_{\ell=1}^{Q}|w_\ell|.
  \label{eq:sisc-lcu-normalization}
\end{equation}
The normalization \(\alpha_Q\) applies to the Duhamel state contribution before receiver evaluation. The calibrated receiver \(\Gamma_h(c_h)\calR_{p_a}\) and the receiver-calibration contribution \(\D\Gamma_h[c_h](\delta c_h)\calR_{p_a}\) require their own preparation, normalization, and selected-observable costs. The finite-dimensional workflow evaluates the two contributions separately and restores their classical scale factors before addition. A circuit-level version can follow the same separate-estimate route: prepare the sparse frozen receiver for the propagated term, query \(\mathcal O_v\) on the receiver support for the calibration term, estimate the corresponding background-pressure functional, and add the two scaled estimates. If \(\alpha_\Gamma\) bounds the calibration-row normalization, the total selected-observable error is the sum of the two estimation errors, and amplification depends jointly on \(\alpha_Q\) and \(\alpha_\Gamma\). A coherent LCU combination requires an additional selector and incurs the combined success-normalization cost.
With qubitization or comparable sparse Hamiltonian simulation, one application of \(U_h(t)\) to accuracy \(\varepsilon_{\rm sim}\) costs \(\widetilde O(\lambda_Ht+\log(1/\varepsilon_{\rm sim}))\) stencil-oracle queries \cite{Berry2015,LowChuang2019}. Before success amplification, state preparation, loading, and measurement, the Hamiltonian-simulation part of the direct construction has query scale
\begin{equation}
  \widetilde O\!\left(Q\,[\lambda_H T
  +\log(Q/\varepsilon_{\rm sim})]\right),
  \label{eq:sisc-born-query-cost}
\end{equation}
together with \(O(Q)\) derivative-block calls. The total complexity must additionally include the implementation and amplification dependence associated with \(\alpha_Q\), source and receiver preparation, perturbation and residual access, and the selected-output observable. Choosing \(Q=\Theta(\lambda_HT)\) for a fixed-order rule makes the displayed Hamiltonian-segment query bound \(\widetilde O((\lambda_HT)^2)\) for the direct construction analyzed here. This estimate does not establish a lower bound for the inverse problem or for all possible Hamiltonian-simulation and quadrature constructions. The LCU success factor, input-loading cost, and measurement cost are accounted for separately. Resolving Hamiltonian-induced oscillations is consistent with analyses of highly oscillatory quantum dynamics \cite{AnFangLin2022}.

Equation~\eqref{eq:sisc-born-query-cost} accounts for the
Hamiltonian-simulation component; Table~\ref{tab:sisc-resource} lists the
additional access, preparation, amplification, and result-extraction
costs.

\begin{proposition}[Entrywise full-gather readout is output-size limited]
\label{prop:sisc-no-full-gather}
If an \(M\)-sample pressure gather is reconstructed by estimating each receiver-time scalar separately to additive tolerance \(\epsilon_{\rm meas}\) and constant success probability, then the standard entrywise procedures use \(O(M/\epsilon_{\rm meas})\) amplitude-estimation calls \cite{Brassard2002} or \(O(M/\epsilon_{\rm meas}^2)\) simple Hadamard-test samples, up to logarithmic amplification factors for simultaneous success over all \(M\) entries. Independently of these estimator upper bounds, explicitly returning \(M\) classical scalars has an \(\Omega(M)\) output-size cost. Thus this entrywise-output model is at least linear in \(M\).
\end{proposition}

\begin{proof}
Estimating all \(M\) entries separately multiplies the scalar-estimation
cost by \(M\), giving the two stated upper bounds. Any explicit classical
representation of the full gather contains \(M\) scalars and therefore has
output-size cost \(\Omega(M)\).
\end{proof}

The proposition is a readout statement for algorithms whose classical output is a full receiver-time gather with entrywise accuracy. Other FWI formulations can use smaller measurement targets, such as scalar misfits, randomized receiver sketches, gradient inner products or Hessian-action probes. The experiments therefore report compressed Hessian-action correlations and finite-shot scalar uncertainty separately from full pressure-trace errors.

\begin{table}[ht]
\centering
\caption{Component-level accounting for the calibrated pressure-observable interface. The middle columns specify access, preparation, normalization, and readout requirements; the final column records the finite-dimensional evidence implemented here.}
\label{tab:sisc-resource}
\scriptsize
\setlength{\tabcolsep}{2.2pt}
\begin{tabular}{@{}L{0.18\textwidth}L{0.27\textwidth}L{0.26\textwidth}L{0.20\textwidth}@{}}
\toprule
Primitive & Input or oracle model & Cost or bottleneck & Evidence in this paper\\
\midrule
Hamiltonian \(H_h(c_h)\) & stencil location oracle and \(\mathcal O_c\) & normalization \(\lambda_H\); boundary-aware construction & matrices, statevectors, and compiled product formula\\
Derivative \(\D H_h(\delta c_h)\) & perturbation value oracle \(\mathcal O_v\), or a separately costed state reduction & normalization \(\lambda_{\D H}\); input conversion & Born actions and compiled LCU zero block\\
Sources/receivers & sparse-support rotations \(U_s,U_\eta\) and stored scales & \(O(k_s+k_\eta)\) support-dependent preparation; background refresh & exact arrays and structured gate preparation\\
Born/adj./Hessian & \(Q\)-node LCU, reverse receiver preparation, and separate calibration-row estimate & Hamiltonian cost \eqref{eq:sisc-born-query-cost}; \(O(Q)\) derivative calls; \(\alpha_Q,\alpha_\Gamma\) amplification & dense/statevector agreement and one-node compiled Born circuit\\
Residual/readout & \(U_r\), selected pressure functionals, and optional \(U_S\) & loading plus \(M,\epsilon_{\rm meas}\)-dependent readout & Bernoulli-sampled ideal-circuit outputs and sketches\\
\bottomrule
\end{tabular}
\end{table}
\FloatBarrier

\section{Compiled Pressure-Born Circuits and Hybrid Local Inversion}
\label{sec:sisc-gate-level}

\subsection{Compiled pressure--Born circuit}

The small circuit prototype instantiates the calibrated pressure Born
row through structured component circuits.
It uses four periodic one-dimensional spatial points, four auxiliary
\(p_a\) points, and two acoustic components, giving 32 amplitudes on
five system qubits. Three selector qubits and one interferometric ancilla
bring the total to nine qubits. Structured source and receiver states are prepared
with \(X\) and \(R_y\) rotations. For this fixed instance the Pauli
coefficients are precomputed classically,
\begin{equation}
 H_h(c_0)=\sum_{\mu=1}^{K_H}a_\mu P_\mu,\qquad
 \D H_h[c_0](v)=\sum_{\nu=1}^{K_D}b_\nu Q_\nu .
 \label{eq:sisc-circuit-pauli}
\end{equation}
Section~\ref{sec:sisc-measurement-resource} analyzes scalable
sparse-oracle access, whereas the fixed 32-amplitude compiled instance
precomputes its Pauli coefficients and tests the gate-level operator and
readout construction. The resource model accounts for the construction of
the sparse oracles. The single-node midpoint rule used here is the \(Q=1\)
specialization of the discrete Duhamel action in~\eqref{eq:sisc-born-lcu}.

Each Pauli exponential is compiled into basis changes, a CNOT parity
chain, and an \(R_Z\) rotation. With \(r\) product-formula repetitions,
the second-order propagator is
\begin{equation}
 U_{2,r}(t)=
 \left[
 \prod_{\mu=1}^{K_H}e^{-\ii a_\mu P_\mu t/(2r)}
 \prod_{\mu=K_H}^{1}e^{-\ii a_\mu P_\mu t/(2r)}
 \right]^r .
 \label{eq:sisc-product-formula}
\end{equation}
For \(\alpha_v=\sum_\nu|b_\nu|\), a selector state with amplitudes
\(\sqrt{|b_\nu|/\alpha_v}\), phase-aware controlled Pauli operations,
and uncomputation produce a
\textsc{Prepare--Select--Unprepare} unitary \(U_{\D H(v)}\) satisfying
\begin{equation}
 (\langle0|\otimes I)U_{\D H(v)}(|0\rangle\otimes I)
 =\D H_h[c_0](v)/\alpha_v .
 \label{eq:sisc-dh-lcu-block}
\end{equation}
The compiled Born circuit uses the single-node midpoint approximation
\begin{equation}
 -\ii T U_{2,r}(T/2)\D H_h[c_0](v)
 U_{2,r}(T/2)|s\rangle .
 \label{eq:sisc-compiled-midpoint}
\end{equation}
The propagated overlap and the explicit receiver-calibration overlap
are estimated in separate interferometers and combined after restoring
their known scales.

This separate-estimate route gives a circuit-level specialization of
\eqref{eq:sisc-hadamard-variance}. Write
\(B_v=s_{\rm p}a_{\rm p}+s_{\rm c}a_{\rm c}\), where \(a_{\rm p}\) and
\(a_{\rm c}\) are the propagated and receiver-calibration interferometric
expectation values, respectively, and \(s_{\rm p},s_{\rm c}\) contain
the corresponding LCU, quadrature, and pressure scales. For independent shot batches
of sizes \(N_{\rm prop}\) and \(N_{\rm cal}\),
\begin{equation}
 \mathbb E\widehat B_v=B_v,\qquad
 \operatorname{Var}(\widehat B_v)=
  \frac{s_{\rm p}^2(1-a_{\rm p}^2)}{N_{\rm prop}}+
  \frac{s_{\rm c}^2(1-a_{\rm c}^2)}{N_{\rm cal}} .
 \label{eq:sisc-circuit-born-variance}
\end{equation}
Equation~\eqref{eq:sisc-circuit-born-variance} explains both the
\(N^{-1/2}\) sampling scale and the larger relative uncertainty when
the two physical contributions nearly cancel.
Together with Corollary~\ref{cor:sisc-stack-error}, it gives the
circuit specialization
\begin{equation}
 |\widehat B_v-B_v^{\rm phys}|
 \leq |\widehat B_v-B_v^{\rm circ}|
 +|B_v^{\rm circ}-B_v^{\rm dense,mid}|
 +|B_v^{\rm dense,mid}-B_v^{\rm phys}|.
 \label{eq:sisc-circuit-error-split}
\end{equation}
This separates finite-shot, compiled, midpoint-quadrature, and discretization
errors.

\subsection{Finite-shot VQLS local update}

The selected outputs drive a small hybrid inversion. At outer
iteration \(k\), forward circuits provide selected pressure data and
Born circuits provide the entries of \(\widetilde J_k\). Classical
assembly forms the residual and the regularized local system
\begin{equation}
 A_k\delta c_k=b_k,\qquad
 A_k=\widetilde J_k^\top\widetilde J_k+\lambda_k I,\qquad
 b_k=-\widetilde J_k^\top\widetilde r_k .
 \label{eq:sisc-hybrid-normal-system}
\end{equation}
Here tildes denote finite-shot estimates. At every iteration and for every
seed, the same rule sets \(\lambda_k=10^{-4}\max\{\tau_k,10^{-14}\}\),
where \(\tau_k=\operatorname{tr}(\widetilde J_k^\top\widetilde J_k)/4\).
Let \(\bar b_k=b_k/\|b_k\|_2\),
\(\widetilde A_k=A_k/\|A_k\|_2\), and
\[
 \beta_k=
 \big\|\widetilde A_k
 (I-|\bar b_k\rangle\langle\bar b_k|)
 \widetilde A_k\big\|_2 .
\]
The two-qubit VQLS ansatz represents a normalized real
four-component direction \(|x(\theta)\rangle\)
\cite{BravoPrieto2023VQLS}. Its finite-shot objective is
\begin{equation}
 C_k(\theta)=
 \frac{\langle x(\theta)|\widetilde A_k
 (I-|\bar b_k\rangle\langle\bar b_k|)
 \widetilde A_k|x(\theta)\rangle}{\beta_k}.
 \label{eq:sisc-vqls-cost}
\end{equation}
The three-parameter ansatz applies \(R_y(\theta_0)\) to the first qubit
and two complementary controlled-\(R_y\) rotations with angles
\(\theta_1,\theta_2\). Classical SPSA \cite{Spall1992SPSA} minimizes the ten-term Pauli
expansion of \eqref{eq:sisc-vqls-cost}; its identity contribution is
known analytically, and the nine nonidentity terms are measured. The
optimizer runs for 1,000 iterations with
\(a_{\rm SPSA}=0.5\), \(c_{\rm SPSA}=0.12\), and 5,000 shots per term.
Two cost evaluations per iteration give
\(2\times1{,}000\times9=18{,}000\) sampled cost circuits; validation, scale
recovery, and signed-amplitude readout are counted separately.

Scale recovery converts the normalized direction into a model step.
With \(\rho_k\) denoting the recovered scale,
\begin{equation}
 \rho_k=
 \|b_k\|_2
 \frac{\operatorname{Re}\langle\bar b_k|A_k|x(\theta)\rangle}
 {\langle x(\theta)|A_k^2|x(\theta)\rangle},
 \qquad
 \widehat{\delta c}_k=\rho_k\,\widehat x(\theta).
 \label{eq:sisc-vqls-scale}
\end{equation}
If \(b_k=0\), the local step is set to zero. Otherwise, each reported
regularized \(A_k\) is positive definite, so the displayed
normalizations and scale denominator are positive.
Finite-shot overlap circuits estimate the scale, and four reference
interferometers recover the signed real components of
\(\widehat x(\theta)\). Validation, scale, and signed-amplitude overlaps
use \(5\times10^4\) shots each.

We draw independent Bernoulli samples from the explicit ideal-circuit
probabilities. These samples have the same distribution as ideal
finite-shot execution of the compiled measurement circuits. Together
with the sampled circuit outputs, classical Jacobian-table assembly,
normal-system construction, Pauli decomposition, SPSA, line search, and
model refresh complete the hybrid local-inversion loop.

Shot counts retain all Bernoulli records from the Hadamard or interference
circuits. The selector-zero probability enters the known LCU scale restored
by the estimator, so \(10^5\) shots denotes the complete set of \(10^5\) raw
circuit samples.

One update comprises a selected forward evaluation, \(P\) Born columns,
classical system assembly, VQLS cost evaluation, scale recovery,
signed-amplitude readout, and repeated forward evaluations for the line
search. The Supplement gives the complete ledger. The parameters
\(P,Q\), product-formula order, conditioning, solver tolerance, and
overlap precision govern the corresponding calls and shot allocation.

\section{Numerical Experiments and Diagnostics}
\label{sec:numerical-evidence}

The experiments track the calibrated pressure observable across
discrete-action and refinement checks, the compiled circuit and hybrid
inversion, and larger statevector and compressed-readout diagnostics.
The Supplement gives the protocols, controls, and input transformations
for the Marmousi-derived array \cite{Brougois1990,Alfarhan2025}.

\subsection{Calibrated acoustic operator checks}

For the implemented finite-dimensional map, the reported Born
finite-difference discrepancy is
\[
 \varepsilon_{\rm BFD}(\epsilon;v)=
 \frac{\|F(c_0+\epsilon v)-F(c_0)-\epsilon Jv\|_2}
 {\epsilon\|Jv\|_2}.
\]
Across three finite-dimensional tests, the Born finite-difference
discrepancy ranges from \(9.33\times10^{-5}\) to
\(5.56\times10^{-4}\), while the adjoint and Gauss--Newton symmetry
defects are at most \(2.70\times10^{-15}\)
(Table~\ref{tab:sisc-main-operator-checks}).
Figure~\ref{fig:sisc-ablation} isolates the receiver-calibration
derivative: omitting it leaves \(O(1)\) finite-difference errors, while
the full \(c\pi\) derivative converges. The table reports relative
Born finite-difference, adjoint-duality, and normal-symmetry defects.

\begin{table}[ht]
\centering
\caption{Main acoustic \(p_a\)-space operator checks at three
finite-dimensional scales. Spatial refinement is reported separately in
Table~\ref{tab:sisc-separated-convergence}. A zero denotes a defect that
vanished in the reported floating-point calculation.}
\label{tab:sisc-main-operator-checks}
\scriptsize
\setlength{\tabcolsep}{3.2pt}
\begin{tabular}{@{}lcccc@{}}
\toprule
Grid/\(N_p\) & Action & Data/model & Born FD (\(\epsilon=10^{-2}\)) & Adjoint/Hessian check\\
\midrule
\(8^2/5\) & matrix & \(15/64\) & \(5.56\times10^{-4}\) & \(0,\;4.07\times10^{-16}\)\\
\(16^2/3\) & matrix & \(24/256\) & \(2.00\times10^{-4}\) & \(2.70\times10^{-15},\;3.10\times10^{-16}\)\\
\(32^2/3\) & matrix-free & \(15/1024\) & \(9.33\times10^{-5}\) & \(1.95\times10^{-16},\;4.90\times10^{-16}\)\\
\bottomrule
\end{tabular}
\end{table}

Independently coded tangent and reverse-adjoint RK4 recurrences agree
with autodiff JVP/VJP actions to a maximum relative defect of
\(4.30\times10^{-15}\) across the Born, adjoint, Gauss--Newton,
weighted-adjoint, and symmetry checks
(Table~\ref{tab:sisc-matched-pressure-baseline}). Centered finite
differences are at most \(1.42\times10^{-9}\). Explicit formation of
the \(8^2\) Jacobian also verifies its transpose and normal actions to
\(2.30\times10^{-16}\) and \(2.19\times10^{-16}\), respectively.

\begin{table}[ht]
\centering
\caption{Independent verification of the same discrete calibrated
pressure map. Manual tangent/adjoint RK4 and autodiff evaluate the
Born, adjoint, and Gauss--Newton actions. Entries are maximum relative
defects over three seeded direction/residual pairs.}
\label{tab:sisc-matched-pressure-baseline}
\scriptsize
\setlength{\tabcolsep}{2.5pt}
\begin{tabular}{@{}rrrrrr@{}}
\toprule
Grid & \(Jv\) & \(J^\top r\) & \(J^\top Jv\) & Weighted adjoint & Normal symmetry\\
\midrule
\(8^2\)  & \(1.01\times10^{-16}\) & \(8.19\times10^{-16}\) & \(4.37\times10^{-16}\) & \(9.01\times10^{-16}\) & \(1.57\times10^{-16}\)\\
\(16^2\) & \(6.57\times10^{-17}\) & \(9.53\times10^{-16}\) & \(4.98\times10^{-16}\) & \(7.02\times10^{-16}\) & \(1.96\times10^{-16}\)\\
\(24^2\) & \(1.03\times10^{-16}\) & \(7.73\times10^{-16}\) & \(8.57\times10^{-16}\) & \(4.30\times10^{-15}\) & \(2.81\times10^{-16}\)\\
\(32^2\) & \(2.91\times10^{-17}\) & \(6.11\times10^{-16}\) & \(1.51\times10^{-15}\) & \(8.09\times10^{-16}\) & \(2.82\times10^{-16}\)\\
\bottomrule
\end{tabular}
\end{table}

The separated convergence diagnostics in
Table~\ref{tab:sisc-separated-convergence} summarize implemented
components of Proposition~\ref{prop:sisc-consistency}. For the spatial
rows, a smooth periodic model fixes the physical inputs while the grid
is refined against a \(192^2\) reference. Across \(16^2\)--\(64^2\),
the observed orders are 1.89--2.17 for pressure data and 1.93--2.17 for
the complete Born action. The finite reference can make the finest-pair
estimates slightly exceed two; the coarsest-to-finest fitted orders over
\(16^2\)--\(64^2\), 2.015 and 2.036,
remain consistent with second-order behavior. Omitting the receiver
derivative instead leaves relative Born errors between 1.127 and 1.150.
The auxiliary and quadrature rows are one-component refined-reference
diagnostics. Supplement Section S7 separately compares quadrature rules
at a fixed Hamiltonian.

\begin{table}[ht]
\centering
\caption{Separated discretization diagnostics for the calibrated acoustic
\(p_a\)-space discretization. Spatial rows use a smooth periodic refinement
with fixed physical inputs; the auxiliary and quadrature rows vary one
component against a refined reference.}
\label{tab:sisc-separated-convergence}
\scriptsize
\setlength{\tabcolsep}{3.2pt}
\begin{tabular}{@{}lccc@{}}
\toprule
Diagnostic & Coarse setting/error & Refined setting/error & Observed behavior\\
\midrule
Periodic pressure grid & \(n=16\), \(2.27\times10^{-4}\) & \(n=64\), \(1.39\times10^{-5}\) & order 1.89--2.17\\
Periodic Born grid & \(n=16\), \(3.50\times10^{-2}\) & \(n=64\), \(2.08\times10^{-3}\) & order 1.93--2.17\\
Auxiliary grid & \(N_p=3\), \(1.19\times10^{-2}\) & \(N_p=11\), \(3.60\times10^{-3}\) & \(3.3\times\) reduction; empirical trend\\
Born quadrature & \(Q=1\), \(3.20\times10^{-2}\) & \(Q=8\), \(2.55\times10^{-4}\) & \(126\times\) reduction; empirical trend\\
\bottomrule
\end{tabular}
\end{table}

\begin{figure}[ht]
\centering
\includegraphics[width=0.96\textwidth]{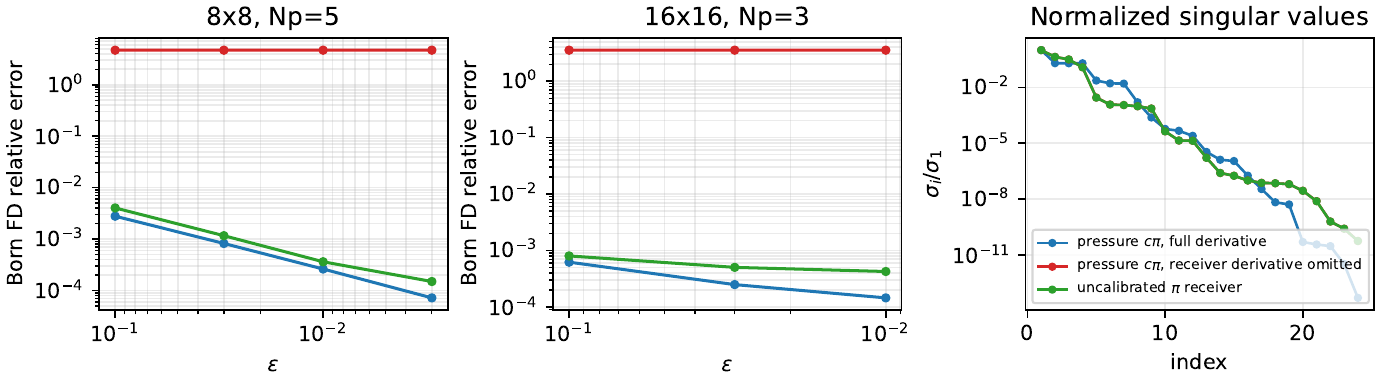}
\caption{Receiver-observable ablation. The physical \(c\pi\) observable with the receiver-calibration derivative passes the Born finite-difference check; omitting the derivative leaves \(O(1)\) errors. The right panel compares normalized singular spectra for the full calibrated pressure row, the same pressure row with its receiver derivative omitted, and the uncalibrated \(\pi\) receiver.}
\label{fig:sisc-ablation}
\end{figure}

Removing the receiver-calibration block yields a regularized
Gauss--Newton step with relative error \(1.296\) and correlation \(0.068\)
against the calibrated step. Its full-data quadratic reduction is
\(-2.101\), and the calibrated-objective line search selects \(\alpha=0\);
the complete step instead achieves a relative reduction of \(0.979\) in
the nonlinear pressure objective. Supplement Section~S7 gives
the complete diagnostics.

\subsection{Compiled circuit and finite-shot hybrid inversion}

The compiled circuit reproduces the dense midpoint Born value to
\(1.89\times10^{-6}\) and its centered finite difference to
\(2.80\times10^{-4}\). The omitted-term ratio
\(\lvert B_{\rm full}-B_{\rm prop}\rvert/\lvert B_{\rm prop}\rvert\) is
\(2.997\). Under the fixed injected first-order source convention, the
prototype uses \(K_H=14\) Hamiltonian and \(K_D=8\) derivative Pauli
terms. The Pauli reconstructions and LCU block agree to at most
\(5.58\times10^{-15}\), and four product-formula repetitions give a
forward discrepancy of \(3.63\times10^{-7}\).
The two contributions are \(-1.430537\times10^{-4}\) and
\(4.287825\times10^{-4}\), giving \(2.857288\times10^{-4}\).
Table~\ref{tab:sisc-gate-hybrid} collects the circuit and inversion
results.

\begin{table}[ht]
\centering
\caption{Compiled pressure-Born and finite-shot hybrid-inversion
summary. Circuit depths and CX counts are basis-transpiled under
unrestricted connectivity. The omitted-term ratio uses
\(\lvert B_{\rm full}-B_{\rm prop}\rvert/\lvert B_{\rm prop}\rvert\),
whereas the circuit--FD discrepancy uses the complete circuit finite
difference as its reference.}
\label{tab:sisc-gate-hybrid}
\scriptsize
\setlength{\tabcolsep}{2.0pt}
\begin{tabular}{@{}L{0.24\textwidth}L{0.20\textwidth}L{0.29\textwidth}L{0.18\textwidth}@{}}
\toprule
Circuit quantity & Value & Hybrid-inversion quantity & Value\\
\midrule
System/selector/readout qubits & \(5/3/1\) & Pressure observations/model parameters & \(32/4\)\\
Forward depth/CX & \(7{,}258/3{,}295\) & Forward/Born shots per overlap & \(10^5\)\\
Born-LCU depth/CX & \(20{,}334/9{,}373\) & VQLS shots per Pauli cost & \(5\times10^3\)\\
Forward/Born relative error & \(3.63\!\times\!10^{-7}/1.89\!\times\!10^{-6}\) & Initial relative model error & \(1.0680\times10^{-1}\)\\
Complete circuit--FD relative discrepancy & \(2.80\times10^{-4}\) & Final error, mean/median & \((2.8069/2.5059)\times10^{-3}\)\\
Omitted-term ratio & \(2.997\) & Runs improving the model & \(10/10\)\\
LCU selector-zero probability & \(0.09050\) & Calls/raw samples per update & \(19{,}572/2.006\times10^8\)\\
\bottomrule
\end{tabular}
\end{table}

The fixed four-parameter test uses 32 synthetic pressure
observations generated by the same forward model, ten predeclared seeds
\(46200\)--\(46209\), at most three outer iterations, and noiseless data.
Each selected forward or Born overlap uses
\(10^5\) shots; the VQLS schedule is specified in
Section~\ref{sec:sisc-gate-level}. Exact statevectors provide the
circuit probabilities and post-run diagnostics, while the inversion uses
the sampled cost, scale, and solution-direction estimators.

For the same frozen initial system \(A_0,b_0\), the direct solve and
zero-start conjugate gradient~\cite{HestenesStiefel1952} provide
classical references. CG reaches a machine-precision residual in four
iterations, consistent with exact-arithmetic termination in at most
four steps for this \(4\times4\) positive-definite system. The direct
solve also reaches a machine-precision residual. The archived
exact-state VQLS reaches \(7.90\times10^{-9}\), while ten finite-shot
VQLS replicates have mean residual 0.1119 at
\(1.362\times10^8\) raw samples per replicate. Because all four solvers
use the same frozen system, the comparison separates algebraic,
variational, and sampling errors while recording measurement work rather
than runtime scaling. The Supplement gives the complete per-update
measurement ledger.

All ten full finite-shot runs reduce the initial \(1.068\times10^{-1}\)
model error to \(4.54\times10^{-4}\)--\(7.26\times10^{-3}\) (sample
standard deviation \(1.95\times10^{-3}\)). With ideal forward/Born
probabilities but the same finite-shot VQLS, the matched control reaches
mean/median final errors \((2.8678/2.1992)\times10^{-3}\), with positive
late-step selections throughout. This comparison helps separate forward/Born
sampling from the retained VQLS sampling and variational errors. With ten
runs per group, the two summary statistics give opposite orderings: the
ideal-forward/Born control has a slightly higher mean but a lower median.
Under the full protocol, a candidate is
accepted when three fresh forward estimates average below the incumbent;
the rule assigns \(\alpha=0\) to two late proposals. The Supplement reports
complete correlations and trajectories; Figure~\ref{fig:sisc-finite-shot-vqls}
shows the distributions.

\begin{figure}[ht]
\centering
\includegraphics[width=0.96\textwidth]{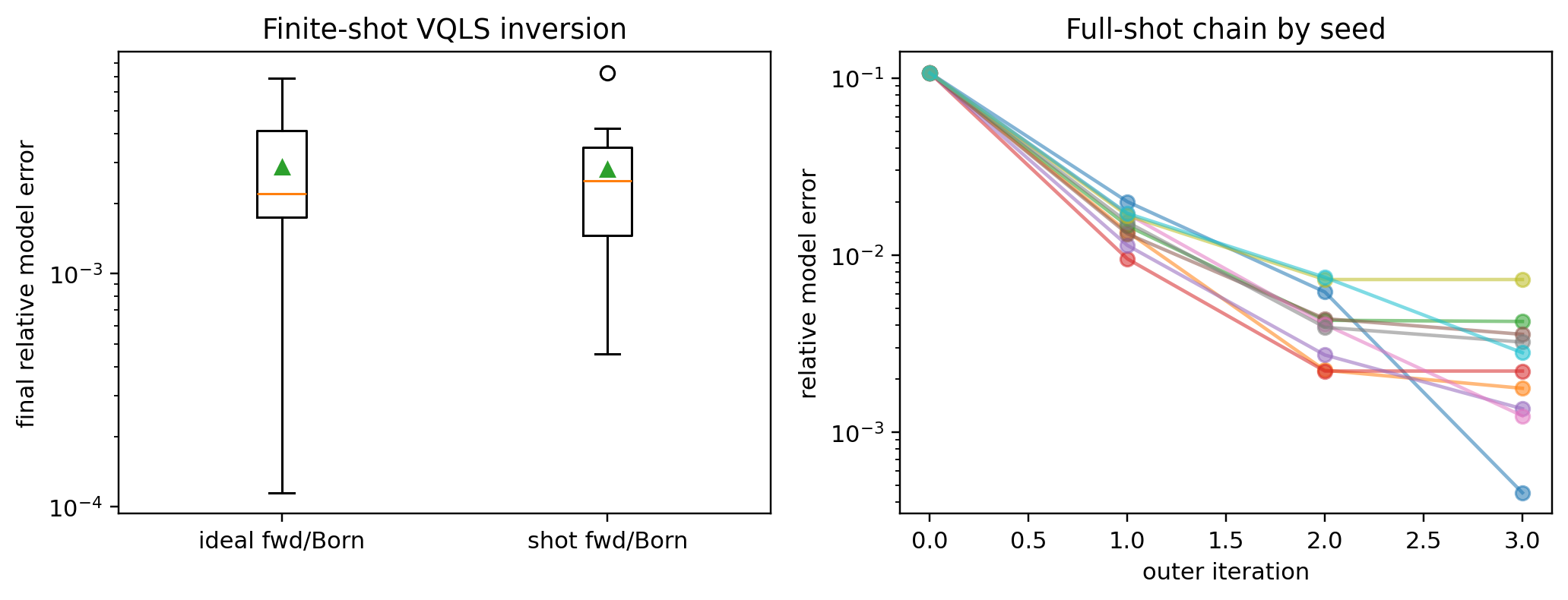}
\caption{Four-parameter hybrid inversion over ten predeclared seeds.
Left: final errors using ideal or Bernoulli-sampled forward/Born
probabilities, both with finite-shot VQLS (triangles: means; lines:
medians). Right: the full finite-shot trajectories.}
\label{fig:sisc-finite-shot-vqls}
\end{figure}

\subsection{Assembled-matrix statevector agreement}

Dense and Qiskit~\cite{JavadiAbhari2024Qiskit}
\texttt{Statevector}/\texttt{Operator} evaluations use identical matrices.
In five of the six archived cases, the discrepancies are recorded explicitly
and remain at the level of double-precision matrix arithmetic; for the largest
Marmousi case, the archive records only a pass against a \(10^{-14}\)
tolerance. Table~\ref{tab:sisc-qiskit-stack} lists representative rows, and
Supplementary Table~S7 reports all six cases. These checks apply assembled
operators to larger states, while the compiled prototype separately supplies
gate decompositions.

\begin{table}[!t]
\centering
\caption{Representative statevector checks; complete results appear in
Supplementary Table~S7.}
\label{tab:sisc-qiskit-stack}
\scriptsize
\setlength{\tabcolsep}{2.6pt}
\renewcommand{\arraystretch}{0.95}
\begin{tabular}{@{}L{0.23\textwidth}L{0.20\textwidth}L{0.25\textwidth}L{0.24\textwidth}@{}}
\toprule
Case & Scale & Agreement result & Role\\
\midrule
\(4^2,N_p=3\) & \(144\to256\) amplitudes; data/model \(24/16\) & Born \(1.17\times10^{-16}\), Hessian \(1.33\times10^{-16}\), RK4 tangent \(3.01\times10^{-5}\) & dense--statevector agreement with RK4 tangent check\\
Geology \(42^2,N_p=3\) & \(15876\to16384\) amplitudes; data/model \(1764/1764\) & forward/Born/Hessian \(6.40\times10^{-17},6.61\times10^{-17},1.43\times10^{-16}\) & structured model statevector agreement\\
Marmousi \(96\times28,N_p=2\) & \(16128\to16384\) amplitudes; 4800 pressure samples & stored tolerance pass (\(10^{-14}\)) & Marmousi backend agreement\\
\bottomrule
\end{tabular}
\end{table}
\FloatBarrier

\subsection{Measurement, sketching, and quadrature}

At retained fraction \(0.5\), ten draws give median update correlations of
\(0.9822\), \(0.9793\), and \(0.9271\) for the Rademacher, Gaussian, and
subsampling sketches, respectively. The corresponding relative reductions
in the nonlinear pressure objective are \(0.9603\), \(0.9652\), and
\(0.9504\). These results quantify update preservation for the reported map,
residual, regularization, and sketch families; the Supplement reports the
full trends.

\section{Discussion}

Receiver calibration belongs to the derivative of the physical data map.
The same energy scaling that enables the Hamiltonian representation defines
pressure as the coefficient-dependent row \(p=c\pi\), whose derivative
contains both \(c_0\delta\pi\) and \(\delta c\,\pi_0\). Omitting the second
term therefore differentiates a different observable rather than a
lower-order approximation to the same one. This distinction persists from
the order-one finite-difference error to the local inverse step: the omitted
map changes the regularized Gauss--Newton direction by relative error
\(1.296\), reduces its correlation with the calibrated direction to
\(0.068\), and yields a direction rejected by the calibrated-objective line
search. The complete direction achieves a relative reduction of \(0.979\)
in the nonlinear pressure objective in the same test.

The numerical evidence separates three implementation levels. Finite
differences, independently coded tangent and reverse-adjoint recurrences,
autodiff actions, and explicit Jacobians verify compatible \(Jv\),
\(J^\top r\), and \(J^\top Jv\) actions for the same finite-dimensional map.
Dense and Qiskit \texttt{Statevector}/\texttt{Operator} evaluations then
apply identical assembled matrices at larger dimensions, measuring backend
agreement rather than a separate discretization or gate synthesis. Finally,
the compiled small instance tests structured preparation, product-formula
propagation, derivative-LCU insertion, separate measurement of the propagated
and receiver-calibration terms, and restoration of their physical scales.
The finite-shot inversion tests how these measured quantities enter a complete
local update.

The readout results clarify the intended output of the quantum subroutine.
Entrywise reconstruction of a complete receiver--time gather scales at least
linearly with the number of reported values, whereas a local update depends on
a smaller prescribed collection or sketch of observables. For the tested
maps, half-size Rademacher and Gaussian sketches retain median update
correlations near \(0.98\) and preserve most of the nonlinear pressure
objective decrease. These empirical results apply to the reported map,
residual, regularization, and sketch families and support selected local-update
observables as a practical measurement target for the reported setting.

Opposing propagated and calibration terms increase relative sampling
uncertainty, while the conservative line search compares late improvements
against fresh noisy estimates. Both effects are consistent with the two
observed zero-step decisions, while the present runs leave their individual
contributions unresolved.

The periodic setting provides a clean finite-dimensional specialization in
which the centered spatial block is exactly skew-adjoint.
The mathematical and numerical claims concern a finite-dimensional,
constant-density, fixed-interface problem. The compiled circuit uses a fixed
32-amplitude instance,
classically precomputed Pauli coefficients, ideal circuit probabilities,
unrestricted connectivity, and Bernoulli sampling. The resource model specifies
scalable state preparation, coefficient access, block encoding, and
amplification; device-noise analysis is deferred to hardware-mapped
implementations. The four-parameter VQLS instance exposes the complete
measurement and update chain. Runtime and asymptotic advantage require separate
scaling studies. Variable density, coefficient-dependent damping or
transducers, hardware-mapped circuits, and explicitly augmented PML
formulations require corresponding derivative, stability, and resource
analyses.

\section{Conclusion}

We constructed a pressure-consistent operator-and-readout interface from
Schr\"odingerised acoustic propagation to the Born, adjoint, and
Gauss--Newton actions required by a local FWI step. The energy scaling that
produces the Hamiltonian representation also makes pressure coefficient
dependent: differentiating \(p=c\pi\) gives \(c_0\delta\pi\) and the direct
receiver-calibration term \(\delta c\,\pi_0\). Retaining both yields
compatible \(Jv\), \(J^\top r\), and \(J^\top Jv\) actions; omitting
calibration changes the observable and can qualitatively alter the update.

Conditional analysis, periodic second-order refinement, independent
discrete-action checks, statevector comparisons, and a compiled pressure--Born
circuit verify complementary parts of the construction. Finite-shot outputs
drive a four-parameter hybrid update, while the resource analysis makes the
access, normalization, and selected-measurement assumptions explicit.
Connecting Hamiltonian wave propagation to local FWI therefore requires a
consistent derivative of the physical receiver and extraction of the update
observables at an accounted measurement cost. Scaling and hardware studies
are needed to assess quantum speedup and device-level performance.

\section*{Data and Code Availability}

A versioned reviewer archive containing analysis scripts, source snapshots,
settings, seeds, logs, machine-readable summaries, environment records, and
file digests is available to the editors and referees on request. The
Supplementary Materials record the public source and checksum of the excluded
third-party Marmousi-derived input. A permanent public repository with a
persistent identifier will follow initial submission; no reproducibility badge
is requested for the initial submission.

\section*{Funding and Acknowledgments}

Computing resources were supported in part by Shenzhen Loop Area Institute (SLAI).

\section*{Declarations}

\paragraph{Competing interests}
The authors declare no competing interests.

\paragraph{Use of AI tools}
During manuscript preparation, the authors used AI-assisted tools for language polishing, structural feedback, code review, and local verification-script generation. The authors independently checked the mathematical statements, references, code, numerical outputs, and figures. The authors assume responsibility for all content.

\bibliographystyle{siamplain}
\fontsize{8}{8.1}\selectfont
\bibliography{references}

\end{document}


\renewcommand{\thetable}{S\arabic{table}}
\renewcommand{\thefigure}{S\arabic{figure}}
\maketitle

\section*{S1. Supplementary evidence and organization}

This supplement provides the numerical protocols and extended evidence
for the main manuscript. Sections S2--S4 document the compiled
pressure-Born circuit, finite-shot local updates, robustness controls,
and the four-parameter \(4\times4\) block reconstruction. Sections
S5--S7 report identical-matrix statevector checks, independent
discrete-action verification, convergence, measurement, sketching, and
quadrature diagnostics. Section S8 gives compact classical reference
studies, and Section S9 maps each reported result to its archived run.
The theoretical statements remain in the main manuscript.
References in this supplement are numbered independently of those in the
main manuscript.

The central calculation uses the constant-density acoustic
\(\pa\)-space Hamiltonian with calibrated \(c\pi\) pressure rows, the
receiver-calibration derivative, and compatible Born, adjoint, and
Gauss--Newton actions. The compiled small instance realizes structured
preparation, product-formula propagation, a
\textsc{Prepare--Select--Unprepare} derivative LCU, and pressure-overlap
measurements. The larger Qiskit cases apply identical
classically assembled matrices through
\texttt{Statevector}/\texttt{Operator} and serve as statevector
agreement checks.

As in the main manuscript, a \emph{selected-output measurement} is a
prescribed subset or linear sketch of receiver--time pressure
functionals used to form the local update. It is distinct from the LCU
\textsc{Select} unitary and selector-zero block.

\paragraph{Error and correlation conventions}
Unless a table states otherwise, vector relative errors use the
Euclidean norm of the named reference, linear residuals use
\(\|A\widehat x-b\|_2/\|b\|_2\), and correlations are Euclidean
cosines. The compiled-circuit omission ratio in S2 normalizes the
missing receiver contribution by the omitted-term response; the smooth
spatial omission study in S6 normalizes by the complete refined
reference. These two ratios test the same missing term under different
reported normalizations and are labeled separately.

\begin{table}[H]
\centering
\scriptsize
\caption{Map of the supplementary evidence. The table records the
quantity tested, the role of each diagnostic, and its scope.}
\label{tab:supp-diagnostic-inventory}
\begin{tabular}{@{}L{0.22\textwidth}L{0.30\textwidth}L{0.21\textwidth}L{0.18\textwidth}@{}}
\toprule
Diagnostic block & Quantity tested & Role & Scope \\
\midrule
Calibrated acoustic checks & Born finite differences, adjoint identities, Hessian actions, receiver-calibration derivative & Main operator verification & Finite-dimensional fixed-interface or pixel derivative \\
Compiled pressure-Born circuit & Structured preparation, product-formula propagation, derivative LCU, calibrated receiver, and selected-output measurement & Gate-level implementation check & 5 system qubits; 9 qubits in the largest interferometer \\
Finite-shot VQLS inversion & Finite-shot forward/Born outputs, VQLS cost and signed-amplitude readout, line search, and model update & Main hybrid calculation & Four parameters, 32 observations, ten predeclared seeds \\
Eight-parameter controls & Shot sweep, receiver-calibration omission, and dense-observation mismatch & Robustness and failure-mode checks & One-dimensional classically orchestrated normal solve \\
\(4\times4\) block reconstruction & Four fixed quadrant parameters on a two-dimensional physical grid & Spatial visualization & Four quadrant parameters \\
Source-convention audit & Fixed physical forcing with and without its source derivative; fixed first-order source states & Born-map definition check & Finite-dimensional forced acoustic system \\
Independent discrete-action verification & Manual RK4 \(Jv\), manual discrete-adjoint \(J^\top r\), composed \(J^\top Jv\), and autodiff on the same \(c\pi\) map & Classical correctness check & float64, multiple directions and grids \\
Separated convergence & Spatial, auxiliary and Duhamel-quadrature components of the consistency estimate & Error-budget diagnostic & clean periodic spatial order plus selected component trends \\
Qiskit pressure-observable calculation & Source state, acoustic \(\pa\)-Hamiltonian, Duhamel Born insertion, frozen calibrated receiver state and receiver derivative & Assembled-matrix statevector agreement & Direct finite-dimensional matrix evolution \\
Finite-shot readout & Bernoulli/Hadamard-style scalar observable noise on stored statevector arrays & Measurement-target diagnostic & Post-processing scalar estimator \\
Compressed sketches & Preservation of sketched gradient and Hessian-action directions & Readout diagnostic & Empirical sketch behavior for tested operators \\
Classical reference studies & CGLS and L-BFGS pressure fits with compact historical context & Classical context & Distinct classical maps and tasks; circuit evidence appears in S2--S7 \\
\bottomrule
\end{tabular}
\end{table}

\section*{S2. Compiled pressure-Born circuit and finite-shot hybrid inversion}

The compiled prototype uses four periodic spatial points, four
auxiliary \(\pa\) points, and two acoustic components. The 32 system
amplitudes occupy five qubits; the derivative LCU uses three selector
qubits, and selected-overlap measurement adds one interferometric
ancilla. The largest measured circuit therefore uses nine qubits.
Source and receiver states are prepared with structured \(X\) and
\(R_y\) rotations. An explicit second-order Pauli product formula evolves
the Hamiltonian, and a \textsc{Prepare--Select--Unprepare} block inserts
the directional derivative. Dense matrix exponentials provide
independent reference values for these compiled operations.

For this fixed small instance, the Hamiltonian and derivative Pauli
coefficients are computed classically from the assembled matrices. The
reported depths and CX counts use unrestricted connectivity in the basis
\texttt{rz}, \texttt{sx}, \texttt{x}, and \texttt{cx}; device-specific
mapping will change these counts.
This finite-size Pauli compilation instantiates the propagation,
derivative insertion, receiver preparation, scale restoration, and
measurement components of the main construction. The sparse row/value
oracles in main Section 4 provide the corresponding matrix-free access
model for larger sizes.

\begin{table}[H]
\centering
\scriptsize
\caption{Compiled pressure-Born component checks and representative
resources. The Born reference uses the same single-node midpoint
Duhamel approximation as the circuit. The complete circuit-FD error is
normalized by the complete circuit finite difference; the omitted-term
ratio \(2.997\) is
\(\lvert B_{\rm full}-B_{\rm prop}\rvert/\lvert B_{\rm prop}\rvert\).}
\label{tab:supp-gate-components}
\begin{tabular}{@{}L{0.58\textwidth}r@{}}
\toprule
Quantity & Result\\
\midrule
Hamiltonian/derivative Pauli terms & \(14/8\)\\
Hamiltonian Pauli reconstruction error & \(1.5338\times10^{-16}\)\\
Derivative Pauli reconstruction error & \(8.0257\times10^{-17}\)\\
LCU selector-zero block error & \(5.5775\times10^{-15}\)\\
Forward-state relative error, four product-formula repetitions & \(3.6339\times10^{-7}\)\\
Complete Born relative error versus dense midpoint reference & \(1.8871\times10^{-6}\)\\
Complete derivative error versus circuit finite difference & \(2.80\times10^{-4}\)\\
Omitted-term ratio, \(\lvert B_{\rm full}-B_{\rm prop}\rvert/\lvert B_{\rm prop}\rvert\) & \(2.997\)\\
Derivative-LCU normalization \(\alpha_v\) & \(0.165\)\\
LCU selector-zero probability & \(0.09050\)\\
Forward interferometer: qubits/depth/CX & \(6/7{,}258/3{,}295\)\\
Born-LCU interferometer: qubits/depth/CX & \(9/20{,}334/9{,}373\)\\
\bottomrule
\end{tabular}
\end{table}

The pressure Born value for the reported direction separates into a
propagated contribution \(-1.430537\times10^{-4}\) and a
receiver-calibration contribution \(4.287825\times10^{-4}\), giving
the complete value \(2.857288\times10^{-4}\). Their opposing signs
make this direction sensitive to finite-shot cancellation and motivate
the separate-estimator variance formula in the main text.

The reported shot counts are raw Bernoulli samples from the Hadamard or
interference circuits. The selector-zero probability \(0.09050\) enters
the known LCU normalization restored in the reported estimator. All
Bernoulli measurement records are retained, so this probability is not a
postselection acceptance rate and does not reduce the stated sample count.

The four-parameter inversion uses 32 selected pressure
observations, ten predeclared seeds \(46200\)--\(46209\), and at most
three outer iterations. Each selected forward/Born overlap uses
\(10^5\) shots. The two-qubit VQLS uses \(5{,}000\) shots per Pauli cost
term and \(5\times10^4\) shots for validation, scale recovery, and each
signed real amplitude. Its cost, scale, and update direction are
therefore finite-shot quantities. The VQLS optimizer and readout use
Bernoulli samples; exact statevectors generate ideal probabilities and
serve as post-run diagnostics.
SPSA parameter updates, selected-output table assembly, normal-system
construction, line search, and model refresh are classical.

At every iteration and for every seed, the four-parameter system uses
\[
 \lambda_k=10^{-4}\max\!\left\{
 \frac{\operatorname{tr}(\widetilde J_k^\top\widetilde J_k)}{4},
 10^{-14}\right\}.
\]

The line search uses the current finite-shot forward estimate as
its incumbent loss. Each declared candidate, including \(\alpha=0\), is
then evaluated by the mean of three fresh forward estimates and is
accepted when that mean is smaller than the incumbent. This fixed,
conservative rule records two late \(\alpha=0\) selections below. The
opposing pressure contributions increase relative uncertainty in
cancellation-sensitive directions, and fresh noisy loss estimates make
late improvements harder to distinguish. The two late zero-step decisions
are consistent with their combined effect.

For \(P\) model parameters, \(N_\alpha\) trial step sizes, and
\(N_{\rm rep}\) forward estimates per trial, the per-update workload is
\[
 C_{\rm update}=C_{\rm fwd}+P C_{\rm Born}+C_{\rm assemble}
 +C_{\rm VQLS}+C_{\rm readout}
 +N_\alpha N_{\rm rep}C_{\rm fwd}.
\]
Here \(P\) controls the Born columns; the Duhamel order \(Q\) and
product-formula repetitions control propagation and derivative costs.
The conditioning of \(A_k\), the solver tolerance, and overlap precision
govern VQLS effort and shot allocation. The finite-shot protocol fixes
\(P=4\), \(N_\alpha=5\), and \(N_{\rm rep}=3\).

The archive also contains a fixed-system diagnostic that freezes
the same four-parameter regularized matrix and right-hand side for every
solver. The matrix has \(\kappa(A_0)=1.1074\), and its stored
double-precision direct step is
\([0.059725,0.128832,-0.131283,0.109233]\). Table~\ref{tab:supp-fixed-vqls}
places a zero-start conjugate-gradient solve~\cite{HestenesStiefel1952},
the archived exact-state VQLS, and ten finite-shot VQLS replicates on
this common algebraic problem.

\begin{table}[H]
\centering
\scriptsize
\caption{Matched solver comparison on the same archived \(A_0,b_0\).
The finite-shot row reports means over ten predeclared replicates; step
errors and correlations use the direct step as the common reference.}
\label{tab:supp-fixed-vqls}
\begin{tabular}{@{}L{0.27\textwidth}rrrr@{}}
\toprule
Method & Relative residual & Step error & Correlation & Effort\\
\midrule
Direct dense & \(7.50\times10^{-17}\) & 0 & 1 & one solve\\
CG, zero initial vector & \(1.12\times10^{-16}\) & \(1.43\times10^{-16}\) & 1 & 4 iterations\\
Exact-state VQLS & \(7.90\times10^{-9}\) & \(7.97\times10^{-9}\) & 1 & depth 20; 4 CX\\
Finite-shot VQLS & 0.111928 & 0.112186 & 0.993486 & \(1.362\times10^8\) shots\\
\bottomrule
\end{tabular}
\end{table}

The finite-shot residual has median 0.109793 and range
0.084645--0.156443. The step error has median 0.109934 and range
0.085249--0.155780; the correlation has median 0.993952 and range
0.987856--0.996374. Its frozen-system signed-direction readout error has mean/maximum
\(6.5901\times10^{-3}/1.0788\times10^{-2}\). The matched
comparison isolates algebraic, variational, and sampling accuracy on one
system, while its workload columns report measurement effort.

Table~\ref{tab:supp-shot-ledger} reconstructs the per-update
ledger. The nine sampled nonidentity cost terms give 18,000 SPSA calls
and 909 validation calls. The frozen-system diagnostic uses six matrix
Pauli terms and 11 scale calls, giving 18,924 total VQLS calls and
\(1.362\times10^8\) shots per replicate. Finite-shot Jacobian
perturbations produce ten matrix terms in the finite-shot outer path; its
scale recovery therefore uses 19 calls.

\begin{table}[H]
\centering
\scriptsize
\caption{Measurement ledger for one four-parameter outer update.
The five line-search candidates are each evaluated three times on 32
selected pressure values.}
\label{tab:supp-shot-ledger}
\begin{tabular}{@{}L{0.38\textwidth}rrr@{}}
\toprule
Component & Calls & Shots per call & Equivalent shots\\
\midrule
SPSA cost, \(2{,}000\) evaluations \(\times\) 9 terms & \(18{,}000\) & \(5{,}000\) & \(9.000\times10^7\)\\
Validation, 101 evaluations \(\times\) 9 terms & 909 & \(50{,}000\) & \(4.545\times10^7\)\\
Scale recovery & 19 & \(50{,}000\) & \(9.50\times10^5\)\\
Signed four-component readout & 4 & \(50{,}000\) & \(2.00\times10^5\)\\
Selected forward data & 32 & \(100{,}000\) & \(3.20\times10^6\)\\
Selected Born table & 128 & \(100{,}000\) & \(1.28\times10^7\)\\
Line search, \(5\times3\times32\) & 480 & \(100{,}000\) & \(4.80\times10^7\)\\
\midrule
Total per update & \(19{,}572\) & -- & \(2.006\times10^8\)\\
\bottomrule
\end{tabular}
\end{table}

\begin{table}[H]
\centering
\scriptsize
\caption{Four-parameter finite-shot VQLS inversion. All ten
runs use the same declared acquisition and shot budgets.}
\label{tab:supp-formal-vqls}
\begin{tabular}{@{}L{0.58\textwidth}r@{}}
\toprule
Quantity & Result\\
\midrule
Truth & \([1.04,1.14,0.89,1.08]\)\\
Initial model & \([0.98,1.01,1.02,0.97]\)\\
Initial relative model error & \(1.068004\times10^{-1}\)\\
Final relative error, mean/median & \((2.806901/2.505864)\times10^{-3}\)\\
Final relative error, minimum/maximum & \(4.535616\times10^{-4}/7.264611\times10^{-3}\)\\
Sample standard deviation & \(1.948060\times10^{-3}\)\\
Ideal-forward/Born group, mean/median final error & \((2.867827/2.199169)\times10^{-3}\)\\
Ideal-forward/Born group, minimum/maximum final error & \(1.150672\times10^{-4}/6.867210\times10^{-3}\)\\
Ideal-forward/Born group, late proposals with \(\alpha=0\) & \(0\)\\
Runs improving the initial model & \(10/10\)\\
Update correlation, mean/median/min/max & \(0.887358/0.970229/0.242322/0.996346\)\\
Outer-inversion signed-amplitude readout error, mean/maximum & \(6.7557\times10^{-3}/1.2998\times10^{-2}\)\\
Late line-search proposals with \(\alpha=0\) & \(2\)\\
\bottomrule
\end{tabular}
\end{table}

The minimum update correlation, 0.242322, occurs at the third outer
update of seed 46208. The fixed line search selects \(\alpha=0\), so the
model remains at its preceding value and the run finishes with relative
error \(7.2646\times10^{-3}\), still well below the initial
\(1.0680\times10^{-1}\). The other zero-step selection occurs at the
third update of seed 46203, where the update correlation is 0.949713 and
the final relative error is \(2.2083\times10^{-3}\). Each row is an
outer-update direction; the ten predeclared seeds remain the independent
units for the final-error distribution.

\begin{figure}[H]
\centering
\includegraphics[width=0.82\textwidth]{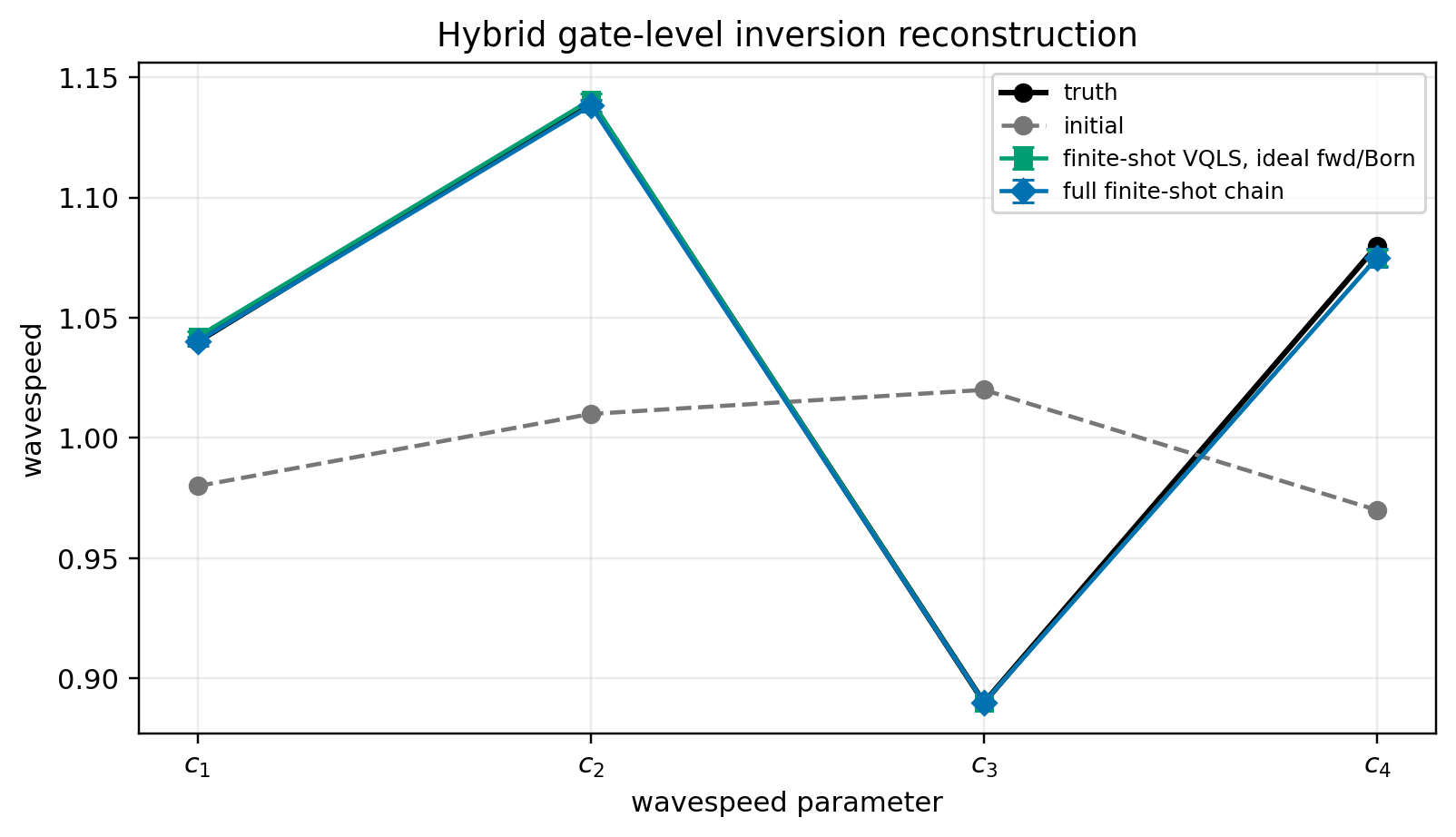}
\caption{Four-parameter hybrid reconstruction. The finite-shot
mean and standard deviation summarize the ten predeclared runs.}
\label{fig:supp-vqls-profile}
\end{figure}

\section*{S3. Eight-parameter shot, calibration, and mismatch controls}

The eight-parameter one-dimensional system is used to test shot scaling
and two controls that would be obscured by the deliberately well
conditioned four-parameter VQLS case. These runs use the same compiled
forward/Born measurement model but solve the small regularized normal
system classically. Every finite-shot row contains ten seeds. The
completion fraction records runs that accept all scheduled updates; an
early \(\alpha=0\) terminates the remaining run and preserves the
improvement already obtained.
The notation T\(r\)/Q\(q\) denotes \(r\) product-formula repetitions
and \(q\) Duhamel quadrature nodes, respectively.

\begin{table}[H]
\centering
\scriptsize
\caption{Eight-parameter reconstruction controls. The dense-observation
rows generate data with a dense exponential and invert with the
product-formula/Born circuit model.}
\label{tab:supp-eight-parameter}
\begin{tabular}{@{}L{0.32\textwidth}rrrr@{}}
\toprule
Protocol & Mean final error & Median final error & Completion & First-step cosine\\
\midrule
Initial model & \(1.01539\times10^{-1}\) & \(1.01539\times10^{-1}\) & -- & --\\
Ideal T4/Q4 & \(1.72563\times10^{-6}\) & \(1.72563\times10^{-6}\) & \(1.0\) & \(1.0\)\\
\(10^4\) shots & \(6.23736\times10^{-2}\) & \(5.76736\times10^{-2}\) & \(0.4\) & \(0.80698\)\\
\(10^5\) shots & \(1.42839\times10^{-2}\) & \(1.32571\times10^{-2}\) & \(0.4\) & \(0.97354\)\\
\(10^6\) shots & \(7.4322\times10^{-3}\) & \(6.9175\times10^{-3}\) & \(0.4\) & \(0.99676\)\\
Dense observations, ideal inversion & \(2.23251\times10^{-5}\) & \(2.23251\times10^{-5}\) & \(1.0\) & --\\
Dense observations, \(10^5\) shots & \(1.84831\times10^{-2}\) & \(2.10497\times10^{-2}\) & \(0.7\) & \(0.97177\)\\
Calibration omitted, \(10^5\) shots & \(1.01539\times10^{-1}\) & \(1.01539\times10^{-1}\) & \(0.0\) & \(0.9927^{\rm omit}\)\\
\bottomrule
\end{tabular}
\end{table}

Increasing the shot budget improves both the reconstruction error and
the first-step cosine, while the completion fraction remains \(0.4\)
for all three shot levels. In the dense-observation test, the
observation generator differs from the inversion propagator, making this a
discretization-mismatch control. For the omitted row, the superscript indicates that 0.9927
compares each finite-shot omitted-term step with the ideal step from the
same omitted-term map. Against the full calibrated ideal step, the ideal
omitted-term direction instead has cosine \(-0.041856\) and relative
error 6.5871; the ten finite-shot omitted-term directions have mean
cosine \(-0.044980\) and mean relative error 6.4923. All ten line
searches reject the first update. Thus the sampled omitted-term solver
reproduces its own map, while the calibrated pressure-objective line search
assigns its direction a zero step.

\begin{figure}[H]
\centering
\includegraphics[width=0.92\textwidth]{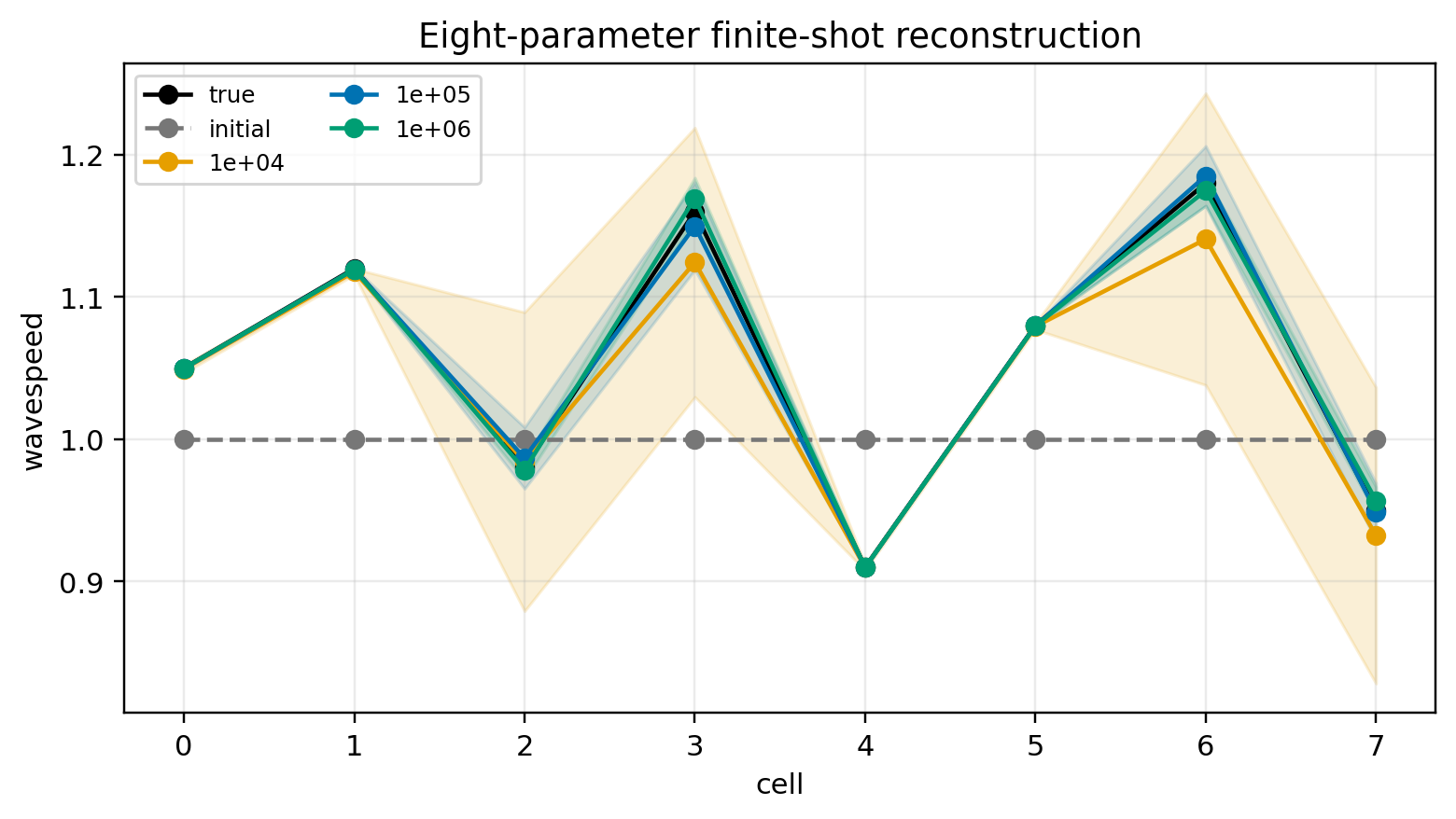}
\caption{Eight-parameter reconstructions across the three shot budgets.
Curves show the mean final model over ten seeds; shaded bands show one
sample standard deviation on either side of the mean.}
\label{fig:supp-eight-profiles}
\end{figure}

\begin{figure}[H]
\centering
\includegraphics[width=0.92\textwidth]{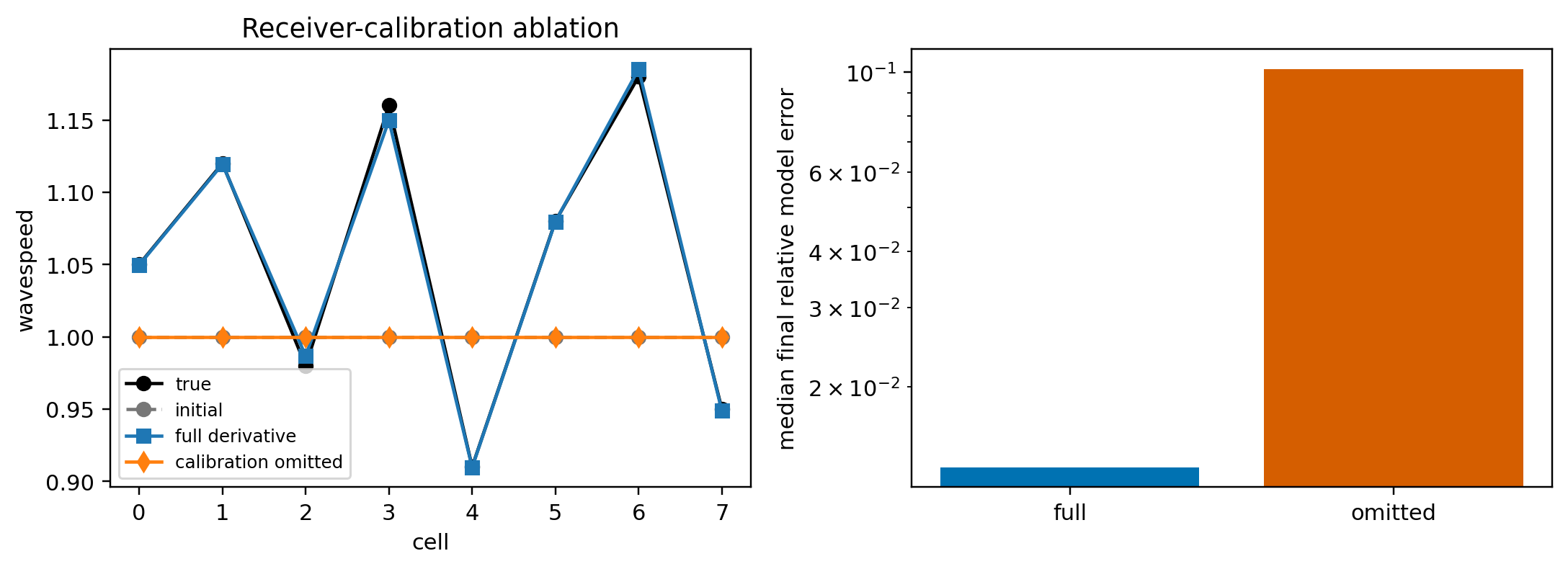}
\caption{Receiver-calibration ablation for the eight-parameter control.
The complete derivative produces accepted updates; the omitted-term
Jacobian is rejected by the calibrated pressure-objective line search.}
\label{fig:supp-calibration-inversion}
\end{figure}

\section*{S4. Four-parameter \(4\times4\) block reconstruction}

The two-dimensional visualization uses a \(4\times4\) physical grid
with four fixed quadrant parameters: top-left, top-right, bottom-left,
and bottom-right. The visualization therefore represents four block
parameters. The ideal T2/Q2 final relative error
is \(2.31554\times10^{-6}\). Across ten \(10^5\)-shot runs, the mean,
median, and sample standard deviation are
\(1.91191\times10^{-2}\), \(1.59314\times10^{-2}\), and
\(1.09661\times10^{-2}\); four runs accept every scheduled update.

\begin{figure}[H]
\centering
\includegraphics[width=0.98\textwidth]{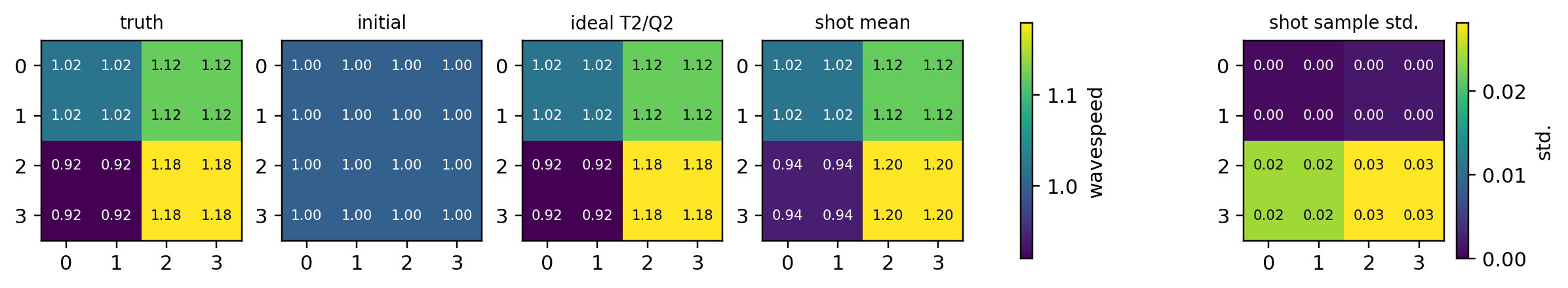}
\caption{Four-parameter block reconstruction on a \(4\times4\)
physical grid. The right panel reports seed-to-seed standard deviation;
repeated values within each quadrant reflect the fixed block
parameterization.}
\label{fig:supp-block-image}
\end{figure}

\section*{S5. Extended calibrated pressure-observable statevector diagnostics}

Table~\ref{tab:supp-qiskit-stack} expands the compressed Qiskit summary
in the main text. All rows in this section use the acoustic
\(\pa\)-space Hamiltonian and calibrated pressure rows. The physical
state dimension is padded to the next power-of-two statevector size. The
dense and Qiskit routes apply identical finite-dimensional matrices, so
their discrepancies measure statevector-backend agreement.
For the \(96\times28\) case, the archived result is a pass against the
\(10^{-14}\) tolerance; the attained discrepancy is unavailable.

\begin{table}[H]
\centering
\scriptsize
\caption{Extended calibrated pressure-observable Qiskit statevector
agreement. Entries report Qiskit-versus-dense discrepancies for identical
assembled matrices or the archived tolerance result. The final column
reports an update-consistency diagnostic for the same discrete map.}
\label{tab:supp-qiskit-stack}
\begin{tabular}{@{}L{0.18\textwidth}L{0.20\textwidth}L{0.26\textwidth}L{0.17\textwidth}L{0.10\textwidth}@{}}
\toprule
Case & State and data scale & Backend agreement & Additional check & SSIM, initial \(\to\) updated \\
\midrule
\(4^2,N_p=3\) & \(144\to256\) amplitudes; data/model \(24/16\) & Born \(1.17\times10^{-16}\), Hessian \(1.33\times10^{-16}\) & RK4 tangent error \(3.01\times10^{-5}\); Born FD \(3.37\times10^{-4}\); adjoint \(1.08\times10^{-15}\) & small agreement case \\
Structured geology \(24^2,N_p=3\) & \(5184\to8192\); \(720/576\) & forward/Born/Hessian \((1.43,1.55,2.53)\times10^{-16}\) & Born FD \(3.16\times10^{-4}\) at \(\epsilon=10^{-2}\) & \(0.709\to0.723\) \\
Structured geology \(42^2,N_p=3\) & \(15876\to16384\); \(1764/1764\) & \(6.40\times10^{-17},6.61\times10^{-17},1.43\times10^{-16}\) & Born FD \(1.11\times10^{-4}\) & \(0.812\to0.817\) \\
Marmousi \(52^2,N_p=2\) & \(16224\to16384\); \(2184/2704\) & \((4.91,4.88,7.21)\times10^{-17}\) & dense/Qiskit agreement & \(0.4561\to0.4695\) \\
Marmousi \(64\times42,N_p=2\) & \(16128\to16384\); \(3072/2688\) & \((4.70,4.44,2.92)\times10^{-17}\) & dense/Qiskit agreement & \(0.4608\to0.4716\) \\
Marmousi \(96\times28,N_p=2\) & \(16128\to16384\); 4800 pressure samples & passed stored \(10^{-14}\) tolerance & Marmousi backend agreement & \(0.5258\to0.5342\) \\
\bottomrule
\end{tabular}
\end{table}

\begin{figure}[H]
\centering
\includegraphics[width=0.95\textwidth]{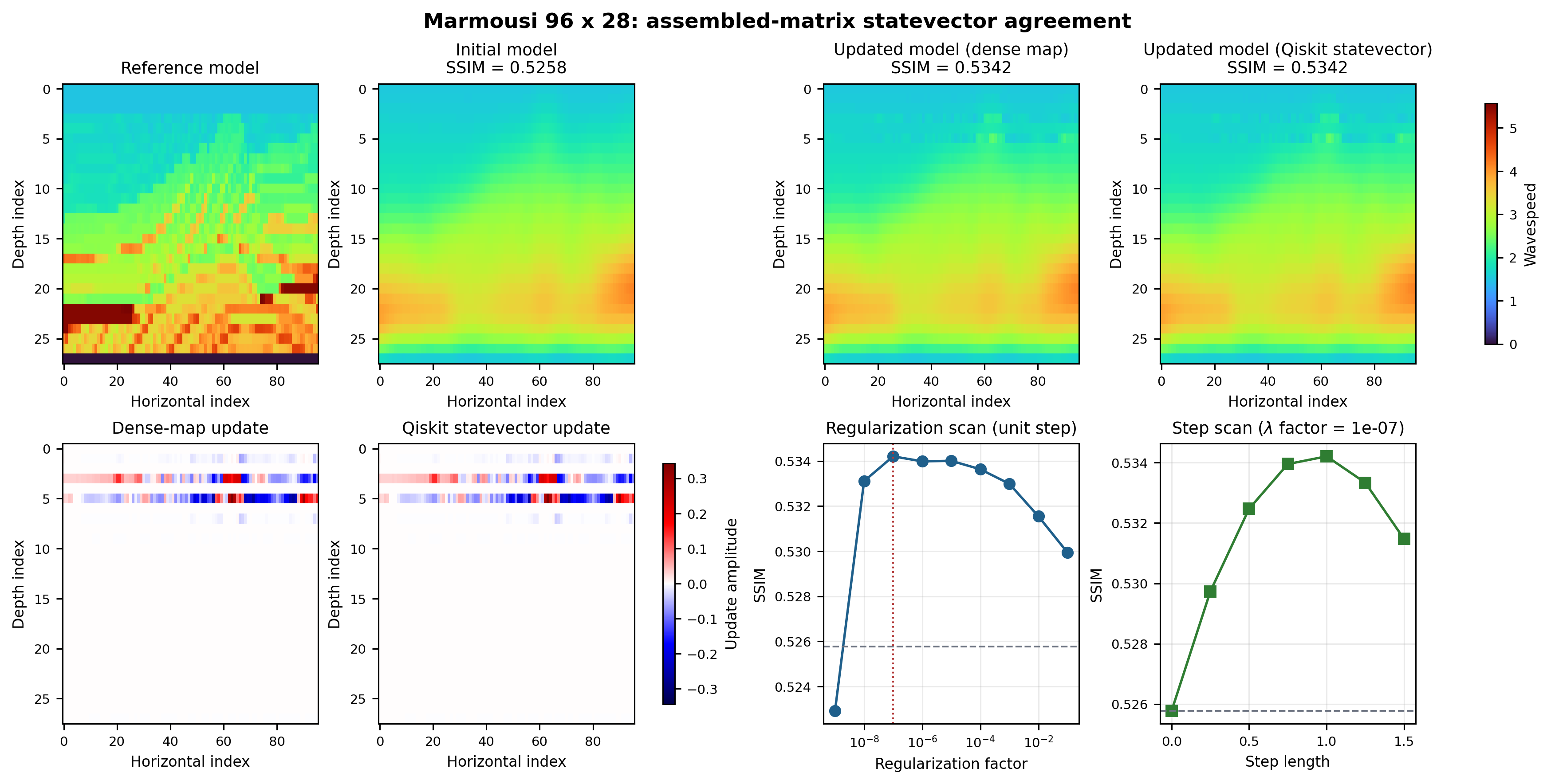}
\caption{Marmousi \(96\times28,N_p=2\) assembled-matrix statevector
agreement for the calibrated pressure observable. Dense and Qiskit
\texttt{Statevector}/\texttt{Operator} apply identical finite-dimensional
Hamiltonian, Duhamel Born, and receiver matrices. The displayed model
updates test backend agreement for this map; the compiled gate-level
prototype is reported separately in Section~S2.}
\label{fig:supp-marm96}
\end{figure}

\section*{S6. Numerical scale and calibrated acoustic checks}

Table~\ref{tab:supp-scales} records the main numerical scales. The
matrix-free rows test the same \(Jv\), \(J^\top r\), and
\(J^\top Jv\) interface at scales where explicit Born-matrix storage is
avoided.

\begin{table}[H]
\centering
\scriptsize
\caption{Numerical scales for the reported calibrated pressure-observable calculations. \(N=n_zn_x\), \(M\) is the number of receiver-time pressure samples, and \(Q\) is the Duhamel quadrature count when an explicit Born matrix is assembled.}
\label{tab:supp-scales}
\begin{tabular}{@{}L{0.24\textwidth}rrrrL{0.30\textwidth}@{}}
\toprule
Diagnostic & \(N\) & \(N_p\) & \(Q\) & \(M\) & State/storage note \\
\midrule
Small Qiskit agreement \(4^2\) & 16 & 3 & 3 & 24 & \(144\to256\) amplitudes \\
Structured Qiskit \(42^2\) & 1764 & 3 & 3 & 1764 & \(15876\to16384\) amplitudes \\
\(52^2\) Marmousi statevector & 2704 & 2 & 3 & 2184 & \(16224\to16384\) amplitudes \\
\(64\times42\) Marmousi statevector & 2688 & 2 & 3 & 3072 & \(16128\to16384\) amplitudes \\
\(96\times28\) Marmousi statevector & 2688 & 2 & 2 & 4800 & \(16128\to16384\) amplitudes \\
\(64^2\) differentiable matrix-free & 4096 & 3 & -- & 1536 & matrix-free JVP/VJP/JTJv actions \\
\(96^2\) matched CGLS & 9216 & 3 & -- & 6912 & augmented CGLS with \(Jv,J^\top r\) \\
\bottomrule
\end{tabular}
\end{table}

\paragraph{Independent verification of the discrete actions}
The stronger float64 check in Table~\ref{tab:supp-independent-all-actions} differentiates the exact discrete RK4 map under fixed coefficient-independent first-order source states. Its tangent and reverse recurrences are coded independently of PyTorch autodiff and include the propagated pressure term and the direct receiver-calibration contribution. Three independently seeded perturbation/residual pairs are used per grid. The weighted-adjoint checks use nonconstant positive diagonal data and model weights. Centered finite-difference errors for \(Jv\) are at most \(1.42\times10^{-9}\). On the \(8^2\) grid, an explicitly formed Jacobian independently gives \(2.30\times10^{-16}\) relative error for the manual VJP versus \(J^\top r\) and \(2.19\times10^{-16}\) for the manual normal action versus \(J^\top Jv\).

\begin{table}[H]
\centering
\scriptsize
\caption{Independent verification of the same discrete calibrated map
in float64. Each entry is the median/maximum relative defect over three
seeded direction/residual pairs.}
\label{tab:supp-independent-all-actions}
\resizebox{\textwidth}{!}{%
\begin{tabular}{@{}rrrrrr@{}}
\toprule
Grid & \(Jv\) & \(J^\top r\) & \(J^\top Jv\) & Weighted adjoint & Normal symmetry\\
\midrule
\(8^2\)  & \(8.08{\times}10^{-17}/1.01{\times}10^{-16}\) & \(5.09{\times}10^{-16}/8.19{\times}10^{-16}\) & \(4.12{\times}10^{-16}/4.37{\times}10^{-16}\) & \(6.90{\times}10^{-16}/9.01{\times}10^{-16}\) & \(0/1.57{\times}10^{-16}\)\\
\(16^2\) & \(6.13{\times}10^{-17}/6.57{\times}10^{-17}\) & \(5.94{\times}10^{-16}/9.53{\times}10^{-16}\) & \(2.89{\times}10^{-16}/4.98{\times}10^{-16}\) & \(3.21{\times}10^{-16}/7.02{\times}10^{-16}\) & \(1.90{\times}10^{-16}/1.96{\times}10^{-16}\)\\
\(24^2\) & \(5.95{\times}10^{-17}/1.03{\times}10^{-16}\) & \(5.98{\times}10^{-16}/7.73{\times}10^{-16}\) & \(5.87{\times}10^{-16}/8.57{\times}10^{-16}\) & \(7.64{\times}10^{-16}/4.30{\times}10^{-15}\) & \(1.50{\times}10^{-16}/2.81{\times}10^{-16}\)\\
\(32^2\) & \(2.28{\times}10^{-17}/2.91{\times}10^{-17}\) & \(3.75{\times}10^{-16}/6.11{\times}10^{-16}\) & \(5.09{\times}10^{-16}/1.51{\times}10^{-15}\) & \(3.52{\times}10^{-16}/8.09{\times}10^{-16}\) & \(1.61{\times}10^{-16}/2.82{\times}10^{-16}\)\\
\bottomrule
\end{tabular}
}
\end{table}

\paragraph{Source-convention audit}
For fixed physical forcing, the first-order source is \(b_s(c,t)=(c f_s(t),0)^\top\), so its Born derivative contains \((\delta c\,f_s,0)^\top\). Table~\ref{tab:supp-source-convention} tests this term on a \(24^2\) forced acoustic system with 80 time steps, 10 recorded times and \(\Delta t=0.01\). The complete fixed-physical-forcing tangent and the fixed first-order-forcing tangent converge under step halving. Omitting the source derivative from the fixed-physical-forcing map leaves a relative error near one. The omitted source term has relative norm 0.499 in the tested tangent action.

\begin{table}[H]
\centering
\scriptsize
\caption{Born finite-difference errors for the two source conventions. The perturbation direction has relative scale 0.02; \(\epsilon\) scales that direction.}
\label{tab:supp-source-convention}
\begin{tabular}{@{}rccc@{}}
\toprule
\(\epsilon\) & Fixed physical forcing, complete tangent & Fixed physical forcing, source term omitted & Fixed first-order forcing\\
\midrule
1       & \(1.342\times10^{-3}\) & 0.9939 & \(1.171\times10^{-4}\)\\
0.5     & \(6.710\times10^{-4}\) & 0.9948 & \(5.858\times10^{-5}\)\\
0.25    & \(3.355\times10^{-4}\) & 0.9952 & \(2.929\times10^{-5}\)\\
0.125   & \(1.677\times10^{-4}\) & 0.9954 & \(1.465\times10^{-5}\)\\
0.0625  & \(8.387\times10^{-5}\) & 0.9955 & \(7.323\times10^{-6}\)\\
\bottomrule
\end{tabular}
\end{table}

\begin{figure}[H]
\centering
\includegraphics[width=0.72\textwidth]{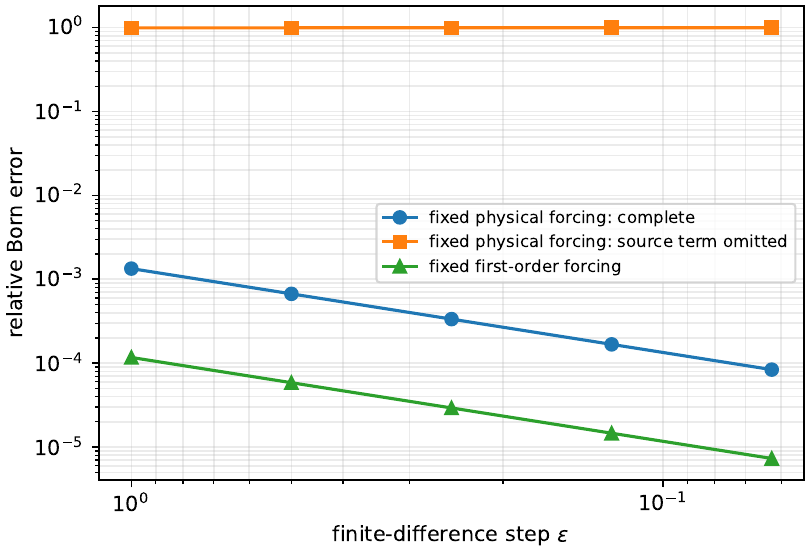}
\caption{Source-convention Born audit. Including the derivative of \(c f_s\)
restores convergence for fixed physical forcing. Holding the injected
first-order source state fixed defines a different, internally consistent
map.}
\label{fig:supp-source-convention}
\end{figure}

\paragraph{Smooth periodic spatial refinement}
Table~\ref{tab:supp-periodic-spatial} isolates the physical-grid component of Proposition~4.2. A smooth low-frequency periodic coefficient, perturbation, and source are used with fixed physical receiver interpolation, recording times, \(N_p=5\), 160 RK4 steps, and \(\Delta t=2.5\times10^{-4}\). Errors are measured against a \(192^2\) reference while the time and auxiliary discretizations remain fixed. The full Born action exhibits the predicted second-order spatial trend, whereas omitting the receiver derivative leaves an \(O(1)\) plateau.

\begin{table}[H]
\centering
\scriptsize
\caption{Smooth periodic physical-grid refinement. Orders use consecutive grid pairs; omitted-row values are relative errors against the complete reference Born action.}
\label{tab:supp-periodic-spatial}
\begin{tabular}{@{}rrrrrr@{}}
\toprule
Grid & Pressure error & Order & Born error & Order & Receiver derivative omitted\\
\midrule
\(16^2\) & \(2.270\times10^{-4}\) & --   & \(3.504\times10^{-2}\) & --   & 1.127\\
\(24^2\) & \(1.057\times10^{-4}\) & 1.89 & \(1.602\times10^{-2}\) & 1.93 & 1.140\\
\(32^2\) & \(5.982\times10^{-5}\) & 1.98 & \(9.009\times10^{-3}\) & 2.00 & 1.145\\
\(48^2\) & \(2.598\times10^{-5}\) & 2.06 & \(3.894\times10^{-3}\) & 2.07 & 1.149\\
\(64^2\) & \(1.392\times10^{-5}\) & 2.17 & \(2.083\times10^{-3}\) & 2.17 & 1.150\\
\bottomrule
\end{tabular}
\end{table}

\begin{table}[H]
\centering
\small
\setlength{\tabcolsep}{2pt}
\caption{Differentiable matrix-free calibrated acoustic \(\pa\)-space
check on a \(64^2,N_p=3\) grid. This implementation reports
matrix-free action accuracy separately from the float64 correctness
table above.}
\label{tab:supp-matrixfree64}
\begin{tabular}{@{}L{0.27\textwidth}L{0.13\textwidth}L{0.13\textwidth}L{0.13\textwidth}L{0.13\textwidth}L{0.15\textwidth}@{}}
\toprule
Grid/backend & Data / model & Born FD & Adjoint & Hessian & Comment \\
\midrule
\(64^2,N_p=3\), differentiable matrix-free & 1536 / 4096 & \(7.51\times10^{-4}\) & \(2.28\times10^{-6}\) & \(5.52\times10^{-5}\) & matrix-free implementation \\
\bottomrule
\end{tabular}
\end{table}

\begin{figure}[H]
\centering
\includegraphics[width=0.94\textwidth]{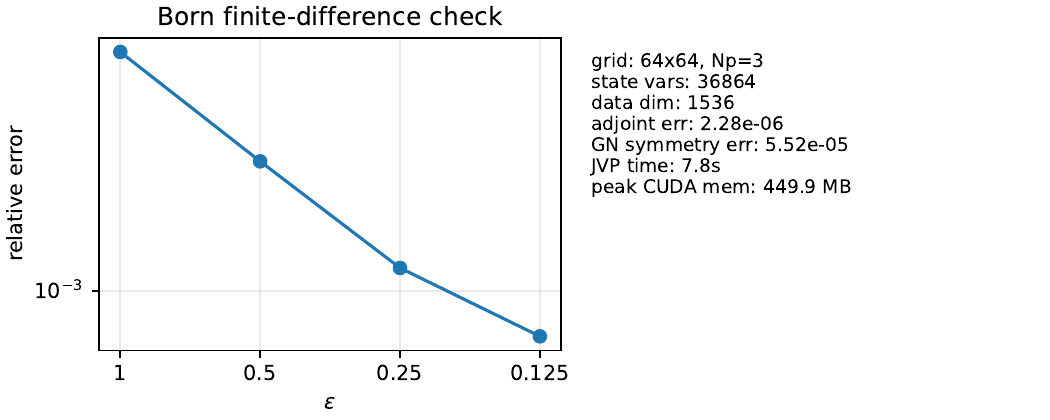}
\caption{Matrix-free calibrated acoustic \(\pa\)-space JVP/VJP diagnostic on a \(64^2\) grid. The check validates the implemented action interface used later by CGLS-style normal equations.}
\label{fig:supp-matrixfree64}
\end{figure}

\section*{S7. Measurement, sketches and Duhamel quadrature}

The main text reports the summary measurement conclusion: selected scalar and sketched Hessian-action observables can preserve local directional information, while entrywise full-gather readout remains output-size limited. Table~\ref{tab:supp-shot-summary} records the largest Bernoulli/Hadamard-style post-processing results used in the main text. The archived generator uses base seed 20260629 (offset by case), 64 independent pseudorandom repetitions, the sample standard deviation with \(\mathrm{ddof}=1\), and the normal-approximation half-width \(1.96s/\sqrt{64}\).

\begin{table}[H]
\centering
\scriptsize
\caption{Bernoulli finite-shot post-processing on stored calibrated \(96\times28,N_p=2\) statevector arrays. Uncertainties are 95\% confidence intervals over 64 repetitions at \(10^6\) shots.}
\label{tab:supp-shot-summary}
\begin{tabular}{@{}L{0.18\textwidth}L{0.19\textwidth}L{0.19\textwidth}L{0.20\textwidth}L{0.17\textwidth}@{}}
\toprule
Case & Born relative error & Gradient relative error & Hessian-action relative error & Hessian-action correlation \\
\midrule
Marmousi \(96\times28,N_p=2\) & \(1.17\times10^{-2}\pm2.64\times10^{-5}\) & \(9.87\times10^{-4}\pm2.21\times10^{-5}\) & \(9.56\times10^{-4}\pm3.03\times10^{-5}\) & \(0.99999954\pm2.94\times10^{-8}\) \\
\bottomrule
\end{tabular}
\end{table}

\begin{table}[H]
\centering
\scriptsize
\caption{Representative compressed Hessian-action sketch correlations
on the \(8^2,N_p=5\) pressure-observable Born matrix with four sources,
eight receivers, six sample times, \(M=192\) data rows, and \(N=64\)
model parameters. These empirical measurement-target diagnostics
accompany the local sketch discussion in the main text.}
\label{tab:supp-sketch}
\begin{tabular}{@{}L{0.15\textwidth}L{0.18\textwidth}L{0.18\textwidth}L{0.38\textwidth}@{}}
\toprule
Sketch rows & Gaussian correlation & Rademacher correlation & Interpretation \\
\midrule
16 & 0.804 & 0.848 & coarse sketch preserves some Hessian-action direction \\
128 & 0.973 & 0.971 & substantially stronger local Hessian-action preservation \\
\bottomrule
\end{tabular}
\end{table}

\begin{table}[H]
\centering
\scriptsize
\caption{Representative Duhamel quadrature scaling. The \(8^2,N_p=5\) calibrated discretization is compared against a 64-node Gauss--Legendre reference with \(\lambda_HT_{\max}=1.388\).}
\label{tab:supp-quadrature}
\begin{tabular}{@{}L{0.13\textwidth}L{0.17\textwidth}L{0.18\textwidth}L{0.17\textwidth}L{0.25\textwidth}@{}}
\toprule
Rule & Low \(Q\) error & Higher \(Q\) error & Reference & Interpretation \\
\midrule
Midpoint & \(3.66\times10^{-3}\) at \(Q=1\) & \(3.45\times10^{-6}\) at \(Q=32\) & 64-node Gauss--Legendre & fixed-order quadrature must resolve Hamiltonian oscillations \\
Gauss--Legendre & -- & \(5.85\times10^{-10}\) at \(Q=4\) & same & higher-order quadrature can be much more efficient in this small test \\
\bottomrule
\end{tabular}
\end{table}

This fixed-Hamiltonian rule comparison is distinct from the
one-component refined-reference quadrature row in main Table 6.3; the
two diagnostics use different configurations and answer different
resolution questions.

\begin{figure}[H]
\centering
\includegraphics[width=0.96\textwidth]{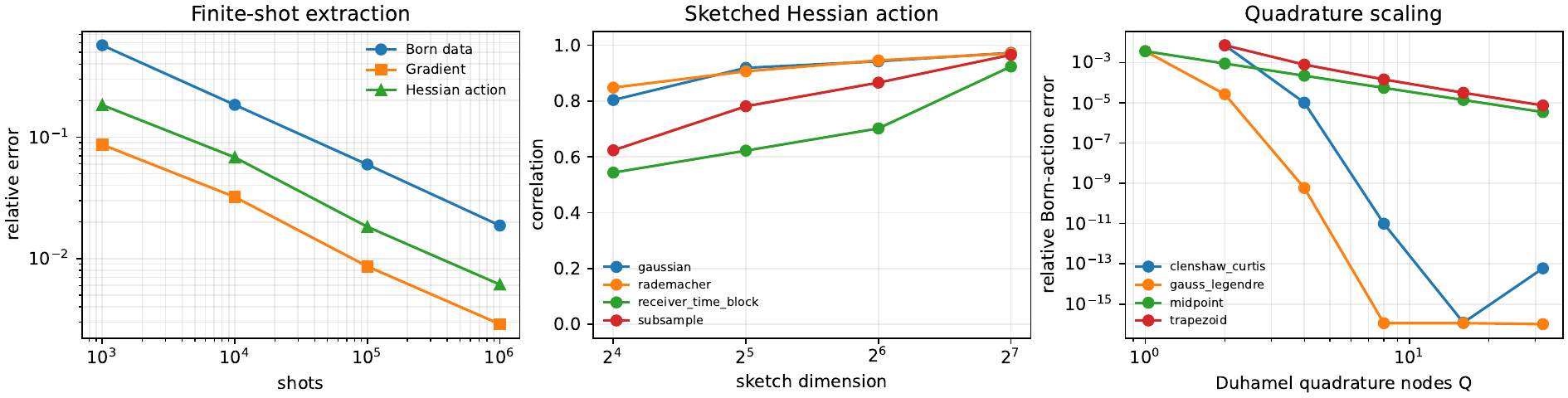}
\caption{Measurement, sketch and Duhamel-quadrature diagnostics for the calibrated pressure-observable calculations. These panels support the result-extraction discussion in the main text and are interpreted through the resource accounting in the main manuscript.}
\label{fig:supp-measurement}
\end{figure}

\paragraph{Compressed local updates for the same discrete map}
Table~\ref{tab:supp-compressed-update} evaluates preservation of the full
regularized local update in addition to a single Hessian-action direction.
The \(16^2\) calibrated map has 100 data samples and 256 model variables.
Each sketch family and dimension uses ten independent draws, and the table
reports medians. The update is compared with the full-data regularized
Gauss--Newton solve; the nonlinear column is the relative reduction in the
full calibrated pressure objective after the stored line search. The
resulting rows quantify local update preservation for the tested map,
residual, regularization, and sketch families.

\begin{table}[H]
\centering
\scriptsize
\caption{Compressed local-update diagnostics for the same discrete map.
Fractions are \(m/M\); all remaining columns are medians over ten
independent sketches.}
\label{tab:supp-compressed-update}
\begin{tabular}{@{}lrrrrr@{}}
\toprule
Sketch & \(m/M\) & Update correlation & Update relative error & Quadratic reduction & Nonlinear pressure reduction\\
\midrule
Rademacher & 0.25 & 0.9589 & 0.2801 & 0.9004 & 0.9393\\
Rademacher & 0.50 & 0.9822 & 0.1810 & 0.9174 & 0.9603\\
Rademacher & 0.75 & 0.9878 & 0.1668 & 0.9235 & 0.9641\\
Gaussian   & 0.25 & 0.9472 & 0.3086 & 0.8937 & 0.9337\\
Gaussian   & 0.50 & 0.9793 & 0.1972 & 0.9220 & 0.9652\\
Gaussian   & 0.75 & 0.9849 & 0.1661 & 0.9240 & 0.9649\\
Subsample  & 0.25 & 0.8115 & 0.5696 & 0.7593 & 0.7941\\
Subsample  & 0.50 & 0.9271 & 0.3603 & 0.9090 & 0.9504\\
Subsample  & 0.75 & 0.9809 & 0.1840 & 0.9269 & 0.9727\\
\bottomrule
\end{tabular}
\end{table}

The direct receiver-calibration block has norm \(1.370\) relative to the
full Jacobian in this test. Removing it changes the regularized
Gauss--Newton step by relative error \(1.296\) and reduces its
correlation with the calibrated step to \(0.068\). The resulting
full-data quadratic reduction is \(-2.101\), and the
calibrated-objective line search selects \(\alpha=0\). By contrast, the
complete calibrated step achieves a relative nonlinear
pressure-objective reduction of \(0.979\). These values quantify the
update-level effect of the receiver derivative in this local setup.

\begin{figure}[H]
\centering
\includegraphics[width=0.96\textwidth]{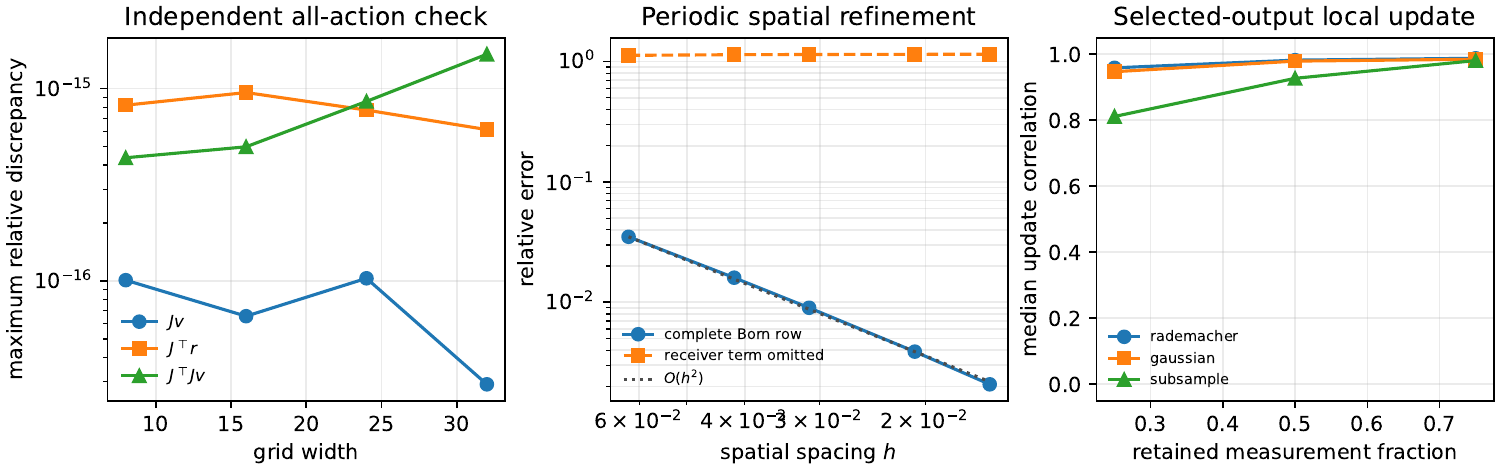}
\caption{Independent discrete-action agreement, smooth periodic
physical-grid convergence, and compressed local-update preservation.
The panels test algebraic correctness, the spatial convergence trend,
and retention of the regularized local update under selected-output
sketches.}
\label{fig:supp-targeted-evidence}
\end{figure}

\section*{S8. Classical reference studies}

This section gives compact classical references for interpreting the
pressure objective and image quality. CGLS
\cite{HestenesStiefel1952,Bjorck1996} consumes \(Jv\) and
\(J^\top r\) through the normal equations, while
L-BFGS~\cite{LiuNocedal1989LBFGS} provides a quasi-Newton pressure-fit
reference. Image quality is summarized by the structural similarity
index (SSIM)~\cite{Wang2004SSIM}. These task-specific discretizations
provide pressure-fit and image-metric references for interpreting the
local results.

\begin{table}[H]
\centering
\scriptsize
\caption{Classical matched-acquisition \(96^2\) Marmousi-derived
direct-action references. Both rows use augmented CGLS with \(Jv\) and
\(J^\top r\) actions; their discretizations are stated separately.}
\label{tab:supp-cgls96}
\begin{tabular}{@{}L{0.25\textwidth}L{0.23\textwidth}L{0.18\textwidth}L{0.20\textwidth}@{}}
\toprule
Grid/model & Operator/backend & Pressure objective & Image-metric outcome \\
\midrule
\(96^2\) Marmousi-derived & calibrated \(\pa\)-space CGLS & \(95.7\times\) reduction & SSIM change \(<4\times10^{-4}\) \\
\(96^2\) Marmousi-derived & unextended acoustic CGLS & \(63.2\times\) reduction & SSIM change \(<4\times10^{-4}\) \\
\bottomrule
\end{tabular}
\end{table}

\begin{table}[H]
\centering
\scriptsize
\caption{Classical acoustic quasi-Newton pressure-fit references on \(96^2\) grids.}
\label{tab:supp-lbfgs}
\begin{tabular}{@{}L{0.20\textwidth}L{0.22\textwidth}L{0.23\textwidth}L{0.18\textwidth}@{}}
\toprule
Case & Backend and settings & Pressure objective & Image-metric outcome \\
\midrule
Marmousi-derived \(96^2\) & same optimizer/control family & \(1.03\times10^{-5}\to2.97\times10^{-8}\) & best SSIM \(0.4552\), final \(0.4434\) \\
\bottomrule
\end{tabular}
\end{table}

The quasi-Newton entries reproduce the archived 12-step histories; the
complete trajectories and image panels remain in the authors'
regeneration archive.

\paragraph{Earlier unpublished internal learned-deblurring study}

Before the compiled circuit and hybrid inversion were available, an
earlier unpublished internal Marmousi study tested whether classically evaluated
pressure-equivalent target actions could train a U-Net
\cite{Ronneberger2015UNet} as a local deblurring preconditioner on the
\(221\times601\) Marmousi array. The target actions, network training,
line search, and model updates were all classical. The archived run
raised SSIM from 0.59294 to 0.68734; a same-acquisition classical
target-action reference reached 0.69959. This study records an earlier
classical use of the pressure-equivalent target action. The operator,
compiled-circuit, and hybrid-inversion results are established by the
experiments in Sections S2--S7.

\section*{S9. Run-to-result index}

Table~\ref{tab:supp-provenance} records the retained Marmousi-derived input and deterministic preprocessing. Table~\ref{tab:supp-repro-index} then maps each numerical claim to the stored summaries and scripts used during manuscript preparation. The archived statevector summaries were generated with Qiskit 2.4.1. A public repository and persistent DOI will follow initial submission.

\begin{table}[H]
\centering
\scriptsize
\caption{Archived model input and deterministic preprocessing used by the reported finite-dimensional diagnostics. File identity is fixed by the recorded dimensions, data type and SHA-256 digest. The Marmousi-derived input is matched to the public DeepFWIHessian repository accompanying Alfarhan et al.}
\label{tab:supp-provenance}
\begin{tabular}{@{}L{0.18\textwidth}L{0.25\textwidth}L{0.47\textwidth}@{}}
\toprule
Model family & Archived input & Selection and preprocessing \\
\midrule
Marmousi-derived & \href{https://github.com/DeepWave-KAUST/DeepFWIHessian/blob/5f9499748bc11f3ac354e1c148eed2f1f17c4e7a/data/Marm.bin}{DeepFWIHessian \texttt{data/Marm.bin}}, commit \texttt{5f9499748bc11f3a} & byte-identical to the public repository file accompanying Alfarhan et al.~\cite{Alfarhan2025}: 132,821 little-endian float32 values in C order, reshaped as \(221\times601\), SHA-256 \texttt{2f399b1a31eab87c}\allowbreak\texttt{f095711a1eb7b09d}\allowbreak\texttt{a9d6eee1bb4ae904}\allowbreak\texttt{5b8dc6e4963cc5a8}; the source paper uses 15 m spacing and cites the original Marmousi model \cite{Brougois1990}; the full array spans 1.484--5.695 km/s and is resized by first-order interpolation \\
Acquisition rule & generated from each target grid & source and receiver depths are stated per experiment; lateral source and receiver coordinates are equally spaced by the deterministic scripts; recorded steps are equally spaced over the simulated time interval \\
\bottomrule
\end{tabular}
\end{table}

\begin{table}[H]
\centering
\scriptsize
\caption{Run-to-result index for the reported numerical evidence.}
\label{tab:supp-repro-index}
\begin{tabular}{@{}L{0.25\textwidth}L{0.31\textwidth}L{0.32\textwidth}@{}}
\toprule
Bundle block & Contents & Tables/figures supported \\
\midrule
Compiled pressure-Born prototype & Structured preparation, 14-term Pauli evolution, 8-term derivative LCU, calibrated readout, finite-shot component checks, and basis-transpiled resources & main circuit section and hybrid summary table; supplement Section S2 \\
Four-parameter finite-shot VQLS & Ten predeclared seeds with finite-shot forward/Born overlaps, VQLS cost, scale, and signed-amplitude readout & main finite-shot inversion figure; supplement Table~\ref{tab:supp-formal-vqls} and Fig.~\ref{fig:supp-vqls-profile} \\
Eight-parameter controls & \(10^4\)--\(10^6\) shot sweep, receiver-calibration omission, and dense-observation mismatch & supplement Section S3 \\
\(4\times4\) block reconstruction & Four fixed quadrant parameters with ideal and ten-seed finite-shot reconstructions & supplement Section S4 and Fig.~\ref{fig:supp-block-image} \\
Calibrated Qiskit statevectors & \(4^2\), \(24^2\), \(42^2\), \(52^2\), \(64\times42\), and \(96\times28\) summaries & main statevector summary table; supplement Section S5 \\
Calibrated acoustic checks & \(8^2\), \(16^2\), matrix-free \(32^2\), differentiable \(64^2\), receiver ablation and separated convergence & main operator and convergence tables; supplement Section S6 \\
Measurement and sketching & Bernoulli statevector post-processing, compressed-observable sketches, and Duhamel-quadrature summaries & main measurement subsection; supplement Section S7 \\
Independent discrete-map evidence & manual tangent/discrete-adjoint and autodiff all-action checks; explicit-Jacobian and finite-difference cross-checks; smooth periodic spatial refinement & main independent-action table; supplement Tables~\ref{tab:supp-independent-all-actions} and \ref{tab:supp-periodic-spatial} \\
Selected-output local update & ten seeded Rademacher, Gaussian and subsampling sketches per setting on the same calibrated map; receiver-calibration update ablation & main measurement subsection; supplement Table~\ref{tab:supp-compressed-update} and Fig.~\ref{fig:supp-targeted-evidence} \\
Classical reference studies & \(96^2\) CGLS and L-BFGS pressure-fit records plus compact historical context & supplement Section S8; full histories remain in the author archive \\
Author-side regeneration archive & source files, bibliography, figure-generation commands, random seeds, JSON/NPZ summaries, SHA-256 manifest, run logs, and environment notes & regeneration record retained by the authors; availability is described in the Data and Code Availability statement \\
\bottomrule
\end{tabular}
\end{table}

\small
\bibliographystyle{siamplain}
\bibliography{references}